\documentclass[12pt]{article}
\usepackage{setspace}
\usepackage{geometry}
\usepackage{hyperref}
\usepackage{enumerate}
\usepackage{paralist}
\usepackage[authoryear]{natbib}
\usepackage{amsmath}
\usepackage{amssymb}
\usepackage{amsfonts}
\usepackage{amsthm}
\usepackage{graphicx}
\usepackage[normalem]{ulem}

\newtheorem {theorem}{Theorem}
\newtheorem {proposition}{Proposition}
\newtheorem {lemma}{Lemma}

\newtheorem {fact }{Fact}

\newtheorem {corollary}{Corollary}

\theoremstyle{remark}

\theoremstyle{definition}
\newtheorem{example}{Example}

\theoremstyle{plain} 
\newtheorem*{definition*}{Definition}
\newtheorem*{theorem*}{Theorem}
\newtheorem*{proposition*}{Proposition}
\newtheorem*{lemma*}{Lemma}
\newtheorem*{claim*}{Claim}
\newtheorem*{subclaim*}{Subclaim}
\newtheorem*{observation*}{Observation}
\newtheorem*{conjecture*}{Conjecture}
\newtheorem*{fact*}{Fact}
\newtheorem*{corollary*}{Corollary}
\newtheorem*{assumption*}{Assumption}
\theoremstyle{remark} 
\newtheorem*{remark*}{Remark}
\newtheorem*{example*}{Example}

\theoremstyle{plain} % just in case the style had changed
\newcommand{\thistheoremname}{}
\newtheorem{genericthm}[theorem]{\thistheoremname}

\newtheorem*{genericthm*}{\thistheoremname}
\newenvironment{namedthm*}[1]
  {\renewcommand{\thistheoremname}{#1}%
   \begin{genericthm*}}
  {\end{genericthm*}}

\newtheoremstyle{mystyle}%                % Name
  {}%                                     % Space above
  {}%                                     % Space below
  {\normalfont}%                                     % Body font
  {}%                                     % Indent amount
  {\bfseries}%                            % Theorem head font
  {.}%                                    % Punctuation after theorem head
  { }%                                    % Space after theorem head, ' ', or \newline
  {\thmname{#1}\thmnumber{ #2}\thmnote{ (#3)}}%                                     % Theorem head spec (can be left empty, meaning `normal')

\theoremstyle{mystyle}

\input{auxiliar/tcilatex}

\usepackage{ifthen}
\newboolean{showcomments}

\newcommand{\kirill}[1]{ \ifthenelse{\boolean{wordcount}}{}{\ifthenelse{\boolean{showcomments}}
{\textcolor{magenta}{(K:  #1)}}{}}{}}

\newboolean{showauxiliar}
\newcommand{\auxiliar}[1]{  \ifthenelse{\boolean{wordcount}}{}{\ifthenelse{\boolean{showauxiliar}}
{#1}{}}{}}

\newboolean{wordcount}
\newcommand{\optional}[1]{\ifthenelse{\boolean{wordcount}}{}{#1}{}}

\usepackage{cleveref}
\makeatletter
\renewenvironment{proof}[1][\proofname] {\par\pushQED{\qed}\normalfont\topsep6\p@\@plus6\p@\relax\trivlist\item[\hskip\labelsep\bfseries#1\@addpunct{.}]\ignorespaces}{\popQED\endtrivlist\@endpefalse}
\makeatother
\crefname{theorem}{theorem}{theorems}
\newenvironment{delayedproof}[1]
 {\begin{proof}[Proof of \Cref{#1}]}
 {\end{proof}}

\usepackage{titlesec}
\usepackage[titletoc,toc,title]{appendix}
\titleformat{\section}
  {\scshape\filcenter}{\thesection.}{1em}{}
\titleformat{\subsection}[runin]
        {\normalfont\bfseries}
        {\thesubsection.}
        {0.5em}
        {}
        [.]
\titleformat{\subsubsection}[runin]
        {\normalfont\bfseries}
        {\thesubsubsection.}
        {0.5em}
        {}
        [.]

\usepackage{tikz}
\usetikzlibrary{calc,positioning,decorations.pathreplacing}

\newcommand\xqed[1]{%
  \leavevmode\unskip\penalty9999 \hbox{}\nobreak\hfill
  \quad\hbox{#1}}
\newcommand\demo{\xqed{$\triangle$}}
\usepackage{mathtools}
\usepackage{blkarray}

\usepackage{blkarray}
\usepackage{mathtools}
\usepackage{colortbl}
\usepackage{bm}

\usepackage{caption}
\usepackage{subcaption}

\usepackage{nicematrix}

\usepackage{xparse}
\usepackage{tikz}
\usetikzlibrary{calc,matrix,patterns,shadings,backgrounds}
\pgfdeclarelayer{myback}
\pgfsetlayers{myback,background,main}

\tikzset{myfillcolor/.style = {draw,fill=#1}}%

\NewDocumentCommand{\highlight}{O{blue!40} m m}{%
\draw[myfillcolor=#1] (#2.north west)rectangle (#3.south east);
}

\NewDocumentCommand{\vshade}{O{blue!40} O{white} m m}{%
\draw[bottom color =#1,top color=#2, draw=none] (#3.north west)rectangle (#4.south east);
}

\NewDocumentCommand{\oshade}{O{blue!40} O{white} m m}{%
\draw[right color =#1,left color=#2, draw=none] (#3.north west)rectangle (#4.south east);
}

\NewDocumentCommand{\inshade}{O{blue!40} O{white} m m}{%
\draw[inner color =#1,outer color=#2, draw=none] (#3.north west)rectangle (#4.south east);
}

\NewDocumentCommand{\fillpattern}{O{north west lines} O{blue!50} m m}{%
\draw[pattern=#1, pattern color=#2, draw=none] (#3.north west)rectangle (#4.south east);
}

\definecolor{princetonorange}{rgb}{1.0, 0.56, 0.0}

\usepackage{titling}
\usepackage{atbegshi}
\usepackage{dsserif}
\DeclareMathAlphabet{\mathpzc}{OT1}{pzc}{m}{it}

\definecolor{orange}{RGB}{238,136,102}
\definecolor{lightblue}{RGB}{119,170,221}
\definecolor{pear}{RGB}{187,204,51}
\definecolor{lightyellow}{RGB}{238,221,136}
\definecolor{pink}{RGB}{255,170,187}
\definecolor{mint}{RGB}{68,187,153}
\newtheoremstyle{special}% name
    {\topsep}%   Space above
    {\topsep}%   Space below
    {\itshape}%  Body font
    {}%          Indent amount
    {\bfseries}% Theorem head font
    {}%          Punctuation after theorem head -- blank
    {0.5em}%     Space after theorem head (0.5em is the default)
    {{\thmname{#1}\thmnumber{ #2$^{\bm*}\!$.}\thmnote{\ \textmd{(#3)}}}}% Theorem head spec 

\theoremstyle{special}

\usepackage[T1]{fontenc}
\usetikzlibrary{arrows.meta}
\usepackage{extarrows}

\usepackage{stackengine}
\stackMath
\usepackage{blkarray, bigstrut}

\usetikzlibrary{arrows.meta,bending}

\makeatletter
\def\th@plain{%
  \thm@notefont{}% same as heading font
  \itshape % body font
}
\def\th@definition{%
  \thm@notefont{}% same as heading font
  \normalfont % body font
}
\makeatother
\usepackage{float}

\usepackage{mathtools}
\DeclarePairedDelimiter\ceil{\lceil}{\rceil}
\DeclarePairedDelimiter\floor{\lfloor}{\rfloor}

\usepackage{afterpage}

\hypersetup {colorlinks=true, linkcolor=orange,
citecolor = orange, urlcolor  = orange}

    \setboolean{wordcount}{false}

        \setboolean{showcomments}{true}
        \setboolean{showauxiliar}{false}

\providecommand{\NOOP}[1]{}

\providecommand{\keywords}[1]{\noindent\textbf{{Keywords:}} #1}

\begin{document}

\newpage

\title{Fragile Stable Matchings}

\author{Kirill Rudov\thanks{\normalsize Department of Economics, UC Berkeley (email: krudov@berkeley.edu). I am grateful to Leeat Yariv for her guidance in developing this project. I would like to thank Noga Alon, Ata Atay, Arjada Bardhi, P{\'e}ter Bir{\'o}, Sylvain Chassang, Federico Echenique, Andrew Ferdowsian, Faruk Gul, Navin Kartik, Alessandro Lizzeri, Sofia Moroni, Xiaosheng Mu, Pietro Ortoleva, Wolfgang Pesendorfer, Doron Ravid, Joseph Root, Alvin Roth, Eran Shmaya, Alex Teytelboym, Can Urgun, Vijay Vazirani, and various seminar and conference participants for helpful comments and insightful discussions. 
}
}
\date{March 14, 2024}

\optional{
\maketitle
\thispagestyle{empty}

\begin{changemargin}{0.25in}{0.25in}
\begin{abstract}

\normalsize  
We show how fragile stable matchings are in a decentralized one-to-one matching setting. The classical work of \cite{roth1990random} suggests simple decentralized dynamics in which randomly-chosen blocking pairs match successively. Such decentralized interactions guarantee convergence to \textit{a} stable matching. Our first theorem shows that, under mild conditions, \textit{any} unstable matching---including a small perturbation of a stable matching---can culminate in \textit{any} stable matching through these dynamics. Our second theorem highlights another aspect of fragility: stabilization may take a long time. Even in markets with a unique stable matching, where the dynamics always converge to the same matching, decentralized interactions can require an exponentially long duration to converge. A small perturbation of a stable matching may lead the market away from stability and involve a sizable proportion of mismatched participants for extended periods. Our results hold for a broad class of dynamics.
\end{abstract}
\end{changemargin}
\clearpage
\setcounter{page}{1}
}

%\tableofcontents
\newpage

\section{Introduction}
\subsection{Overview} We examine the fragility of stable matchings in a decentralized one-to-one matching environment.\footnote{A matching is stable if there is no pair of participants that prefer each other to their partners in the matching; such a pair is called a blocking pair (see \citealp{roth_two-sided_1992} for details).} Most of the existing literature studies centralized markets: the medical residency match, school allocations, and others. However, many markets are not fully centralized: the market for junior economists, college admissions, marriage markets, and so on. Furthermore, decentralized interactions often precede or follow centralized markets. Even when a market relies on a stable centralized clearinghouse, ex-post preference shocks, changes in market composition, or small implementation errors may lead to an unstable matching, close to or distant from the intended stable matching. How fragile are stable matchings with respect to small perturbations? In general, when a matching in place is unstable, what is the set of outcomes decentralized interactions can generate? Even if convergence to a desired stable matching is guaranteed, what are efficiency costs of decentralized interactions, in terms of duration of instability and the proportion of market participants who spend long periods mismatched? These questions are at the heart of this paper.

We consider a decentralized process in which randomly-chosen blocking pairs successively match with each other, potentially breaking their former matches. These dynamics always culminate in a stable matching and impose low sophistication requirements on market participants. Theorem~\ref{theorem_characterization} shows that, under mild conditions, starting from \textit{any} unstable matching, these dynamics attain \textit{any} stable matching with positive probability. In particular, even small perturbations of stable matchings may lead to \textit{any} stable matching. Thus, stable matchings are fragile. Simulations suggest that, in fact, even a minimal deviation from a stable matching might converge with substantial probability to another, possibly vastly different stable matching. Another insight of this paper is that the stabilization dynamics may involve many mismatched participants for long durations. Theorem~\ref{theorem_exponential} shows that, even in markets with a unique stable matching, where the dynamics always converge to the same stable matching, stabilization may take an exponentially long time. This result highlights another fragility aspect: a small perturbation of a stable outcome may lead the market away from stability for a long stretch of time. Using simulations, we show that, whether or not the stable matching is unique, the stabilization process typically strays from stability, takes a long time to regain it, and involves many mismatched participants for substantial periods.

Altogether, our results imply the fragility of stable matchings. Our analysis illustrates the importance of taking into account decentralized interactions that precede or follow centralized markets and casts doubt on empirical studies of decentralized markets that estimate parameters in one snapshot of time, relying on stability as an identification assumption.

Stability is a central concept in matching theory. It is defined by the absence of blocking pairs, which are pairs of participants who prefer each other over their current partners \citep{gale1962college}. This definition is motivated by the concern that some participants may profitably deviate and circumvent the intended matching by forming blocking pairs. Such a concern was part of the impetus for introducing stable centralized clearinghouses.

Folk wisdom suggests that if an unstable outcome occurs, decentralized interactions eventually yield stability. Going back to \cite{knuth1976marriages}, the literature has sought to provide a theoretical framework for understanding how stable outcomes emerge. The classical work of \cite{roth1990random} suggests simple and natural decentralized dynamics generating stability. They show that in one-to-one markets, for any unstable matching, there exists a finite sequence of blocking pairs that generates a stable matching. Consequently, when blocking pairs are formed at random in each step, stability is guaranteed. A benefit of these dynamics is that convergence to stability is possible even when market participants have limited sophistication and possess only local information---participants need to identify their own blocking pairs, but not much more.

We consider a broad class of such dynamics, allowing for non-uniform and time-dependent formation of blocking pairs, for which convergence to stability is guaranteed (cf. \citealp{roth1990random}). This class includes dynamics where, at each step, a participant randomly chosen from those with at least one blocking partner selects her \textit{best} blocking partner. The probability that a particular blocking pair of market participants matches can be quite general, reflecting factors such as the likelihood that these participants would meet, or their incentives to form a match, potentially driven by cardinal payoffs.

We study the robustness of stable matchings with respect to arbitrary, possibly minimal, perturbations. Such perturbations may arise in both decentralized and centralized settings. Over time, participants' preferences may shift, market composition may change, or matched participants might end their partnerships. For instance, in labor markets, a family shock may generate new geographical preferences for workers; new positions may appear and workers may either become available or retire; an employer and an employee might mutually agree to terminate the employment relationship. While the set of stable matchings may also be altered by such changes, they could all lead to instability. In particular, some participants that should be paired under \textit{current} market conditions might be mismatched or unmatched.

In some situations, decentralized interactions cannot break certain partnerships. Imagine a two-sided market divided into two submarkets, New York and Los Angeles, in which every Angeleno prefers any Angeleno to any New Yorker, and vice versa. Suppose that the New York submarket has a unique stable matching and the Los Angeles submarket has two stable matchings. The entire market naturally has two stable matchings as well. If all Angelenos are paired in accordance with one of the two stable matchings for the Los Angeles submarket, and therefore for the entire market, decentralized interactions cannot unmatch any such pair of matched Angelenos. As a result, the stabilization dynamics cannot attain another stable matching in the entire market. Here, Angelenos can be matched in a stable way inside the group so that every insider prefers her stable partner to anyone outside Los Angeles. We call such a group a \textit{fragment}. In this example, New Yorkers form a fragment as well. A fragment is \textit{trivial} if all stable matchings in the entire market match participants inside the fragment in the same way. In particular, Angelenos constitute a \textit{non-trivial} fragment, while New Yorkers form a trivial fragment.

If a market has a non-trivial fragment, some unstable matchings cannot yield certain stable matchings. Theorem~\ref{theorem_characterization} shows that the reverse is also true. When there are no non-trivial fragments, \textit{any} unstable matching can yield \textit{any} stable matching through decentralized interactions. Non-trivial fragments are the only restraints on the stabilization dynamics. The absence of non-trivial fragments is a mild condition: simulations suggest that they are rare in large random markets. Thus, in most markets, the decentralized process is fluid enough to attain any stable matching. 

Theorem~\ref{theorem_characterization} suggests one type of fragility of stable matchings. A small perturbation of a stable matching may culminate in \textit{any} stable matching, close to or distant from the original stable matching. The celebrated Deferred Acceptance mechanism of \cite{gale1962college}, which is used to match doctors to hospitals and students to schools, aims at implementing an extremal stable matching, the best stable matching for one market side, doctors and students, respectively. Our results imply that any perturbation of such a clearinghouse's outcome need not revert back to the intended extremal stable matching; decentralized interactions can instead lead to the other extremal stable matching, or anything in between.

We employ simulations to investigate the stable matchings that have stronger drawing power in random markets with multiple stable matchings. In order to inspect the effects of minimal perturbations to stability, we initialize markets at an \textit{almost stable} matching, which is one blocking pair away from stable. For presentation simplicity, suppose a blocking pair is chosen uniformly at random at each step of the dynamics.\footnote{Analogous results hold for other dynamics, specifically for dynamics where, at each step, a randomly-chosen participant forms a match with her best blocking partner.} The simulations show that almost stable matchings---even those corresponding to an extremal stable matching---converge with substantial probability to a stable matching distant from the slightly perturbed matching.

Fragility may take other forms: the stabilization dynamics may take a long time and involve a sizable fraction of participants mismatched for a long duration. To isolate these other forms of fragility, we focus on markets with a unique stable matching. For such markets, the dynamics always converge to the same matching, the unique stable one. In fact, there is work suggesting that large markets have small cores in certain settings.\footnote{\label{footnote_mixed_evidence}Several empirical papers document small cores for \textit{reported} preferences in some markets (e.g., \citealp{roth1999redesign}). There are also theoretical studies identifying conditions under which large markets have an essentially unique stable matching (see \citealp*{Immorlica2005}, \citealp*{kojima2009incentives}, and \citealp*{ashlagi2017unbalanced}). Of course, our first theorem has bite as long as there are multiple stable matchings, even when they are not numerous. In addition, there is growing evidence that reported preferences might differ significantly from truthful ones; see \cite*{artemov2023stable}, \cite{echenique2022top}, and references therein. Furthermore, recent papers suggest that cores are large in certain settings, including large imbalanced markets (\citealp{biro2022large}; \citealp*{hoffman2023stable, rheingans2024large, brilliantova2022fair}), with possibly correlated preferences.}  

Theorem~\ref{theorem_exponential} shows that for a large class of markets with a unique stable matching, even a small deviation from stability entails an exponentially long stabilization process. Specifically, we start with \textit{any} market with a unique stable matching. We show that the market can be augmented with a small number of market participants so that the resulting market still features a unique stable matching, but where for most matchings on any path to stability, there are many blocking pairs. Importantly, considerably more blocking pairs are ``destabilizing,'' moving the augmented market away from stability, than ``stabilizing,'' moving the augmented market towards stability. The stabilization dynamics can then be associated with a random walk that is heavily biased in the direction of instability. This bias substantially decelerates the convergence, making it exponential.

We illustrate that fragments accelerate convergence. For example, consider a labor market in which participants on each market side have assortative preferences---say, firms have the same ranking of workers based on their college GPA and workers rely on the same company ratings site to rank firms. Then, the best firm and the best worker form a trivial fragment. In addition, the top two firms and the top two workers constitute a trivial fragment that nests the previous one, and so forth. This market belongs to a general class of markets having a nested structure of fragments. We show that for any market in this class, the stabilization dynamics take polynomial time. However, in view of our second theorem, even such markets are fragile when augmented with a small fraction of new participants.

We use simulations to inspect convergence speeds in random markets with and without a unique stable matching. The simulations suggest that, regardless of whether stable matchings are unique, and even when starting from almost stable matchings, the stabilization dynamics typically move far away from stability. Consequently, stabilization takes a very long time, and involves long durations with a sizable fraction of market participants being mismatched.

Taken together, our results suggest the importance of accounting for decentralized interactions even in the presence of stable matching clearinghouses. Small deviations from stability might lead to significant and undesirable changes affecting a large fraction of market participants for a long period of time. From a market design perspective, interventions directed at reducing efficiency costs of decentralized interactions might prove valuable. These interventions may include designing specific rules governing how participants interact with each other or introducing restrictions that reduce market activity, possibly at the expense of flexibility of decentralized interactions to adjust to market conditions. An alternative would be to devise a dynamic reassignment clearinghouse that fixes instabilities on a regular basis; this might be especially relevant given the dynamic nature of this problem. 

Our results also have implications for empirical work on matching. A growing empirical literature uses stability of observed matches as an identification assumption to estimate preferences in decentralized markets (see the surveys by \citealp{fox2009structural} and \citealp{chiappori_econometrics_2016}, as well as \citealp{boyd2013analyzing}). Our findings imply that observed matchings might be far from stable for a long duration. This raises caution for the interpretation of those estimates.

We provide additional implications in the discussion section at the end of the paper.

\subsection{Related Literature}
To the best of our knowledge, no prior work has examined the fragility of stable matchings with respect to arbitrary perturbations through the lens of decentralized interactions. In a somewhat related study, \cite{jackson2002evolution} use similar dynamics to investigate the evolution of networks in the presence of mutations that might randomly add or delete links. Applied to matching problems, their results show that the support of the limiting stationary distribution coincides with the set of stable matchings; in this sense, no stable matching is more ``fragile'' than others in the face of such mutations.\footnote{\cite{newton2015one} show that for alternative mutation models, some stable matchings might not be stochastically stable.} Close to our work, the recent and growing computer science literature studies robust solutions to stable matchings in the presence of incomplete information about preferences \citep{aziz2020stable}, stable matchings that are most tolerant to special classes of errors \citep{genc2019complexity, reddy2022structural}, and algorithms to find a stable matching---after certain changes---that is as close as possible to an original stable matching (see \citealp*{boehmer2022deepening} and references therein), to name a few. More broadly, the current paper relates to the fragility of complex economic systems (e.g., see \citealp{elliott2022networks}).\footnote{\cite{kojima2011robust} examines a mechanism for the school choice setting that is robust to a combined manipulation, where a student first misreports his preferences and then blocks the implemented matching.}

A few papers model a decentralized matching process by which stable outcomes are created. \cite{roth1990random} propose simple dynamics of pairwise blocking.\footnote{\label{footnote_stabilization}Similar stabilization results were later obtained for roommate markets \citep*{chung2000existence, diamantoudi2004random, inarra2008random}, many-to-one markets with couples \citep{klaus2007paths}, many-to-many markets without contracts \citep{kojima2008random} and with contracts \citep{millan2018random}, markets with incomplete information \citep{lazarova2017paths, chen2020learning}, and supply chain networks \citep{rudov2023decentralized}.}\textsuperscript{,}\footnote{Another stream of the literature models market participants as farsighted, see \cite*{herings2020matching} and references therein.} Some recent studies impose additional structure on the matching process or markets to define a decentralized market game in which agents interact strategically; see \cite{haeringer2011decentralized}, \cite*{ferdowsian2023decentralized}, and references therein. They seek to identify conditions under which stable outcomes arise as equilibrium outcomes. We contribute to this literature by shifting the focus to the robustness of stability. As a by-product, we characterize markets in which the \cite{roth1990random} dynamics can yield any stable outcome.\footnote{The search and matching literature pursues a rather different approach to the emergence of stable matchings, in which frictions are modeled explicitly and play a crucial role (see the survey by \citealp*{chade2017sorting} and references therein). Both approaches are used in empirical work to estimate preferences in decentralized markets \citep{chiappori_econometrics_2016, chiappori2020theory}.}

Several experimental papers study frictionless decentralized markets with complete information \citep*{echenique2023experimental, pais2020decentralized}.\footnote{\cite*{pais2020decentralized} also study the impact of commitment, search costs, and limited information. \cite{agranov2023paying} show that incomplete information hinders stability in markets with transfers.} These studies examine convergence to stable matchings and their selection assuming that participants start with an empty matching in which no one is matched. Stable outcomes are prevalent in lab markets. In markets with multiple stable matchings, a non-extremal stable matching---which presents a compromise between the two sides of the market---emerges most frequently \citep*{echenique2023experimental}. For robustness purposes, we consider any starting matching. We show that even a minimal deviation from an extremal stable matching may lead to a completely different stable matching with substantial probability.

Although the convergence dynamics and the selection of stable matchings have recently received attention in the theoretical literature, mostly in computer science, they remain largely understudied. \cite{ackermann2011uncoordinated} construct a particular sequence of markets, with an increasing number of stable matchings, and the corresponding sequence of starting matchings to illustrate an exponential lower bound for the convergence time assuming uniform \cite{roth1990random}-type random dynamics. They also identify a class of markets for which the convergence is polynomial.\footnote{Later, \cite*{hoffman2013jealousy} expanded this class by using a graph-theoretic approach.} Motivated by fragility, we extend their techniques to prove that, for a large class of markets with a unique stable matching, small deviations from stability lead to the exponential stabilization under a wide range of dynamics. We consider a broader class of markets for which the stabilization is polynomial. Moreover, we examine how the stabilization dynamics affect market participants. \cite{biro2013analysis} connect the dynamics to Markov chains and show how to calculate the absorption probabilities of various stable matchings.\footnote{See also \cite{boudreau2011note, boudreau2012exploration}.} We use this connection to quantify the fragility of stable matchings in smaller markets. More generally, the current paper suggests that fragments might be a key determinant of the dynamics. 

Several studies design re-equilibration mechanisms for labor markets in which matchings are destabilized by retirements or new entries. \cite*{blum1997vacancy} propose a procedure, closely related to the deferred acceptance algorithm, that regains stability in case of such disruptions. \cite{combe2022market} introduce static mechanisms involving the assignment of new workers and reassignment of existing ones to rebalance markets that suffer from distributional issues---in their empirical application, the unequal distribution of experienced public school teachers across regions of France.\footnote{See \cite*{combe2022design} and references in these two papers.} Our results indicate that promptly addressing other instabilities might be just as important. Since all such issues are dynamic in their nature, they may require the design of dynamic centralized clearinghouses.\footnote{\cite*{akbarpour2020thickness} study dynamic mechanisms for networked markets in which agents enter or exit stochastically over time. See \cite{baccara2023dynamic} for a review of the recent developments in dynamic matching.}

\section{The Model}
\label{Model}

\subsection{Basic Definitions}

A \textit{matching market} is a triplet $\mathcal{M}=(F,W,\succ)$
composed of a finite set of firms $F=\{f_i\}_{i\in[n]}$, where $[n] = \{1,2,\ldots,n\}$, a finite set of workers $W=\{w_j\}_{j \in [m]}$, and a profile of strict preferences $\succ=\{\succ_{f_{i}},\succ_{w_{j}}\}_{(i,j)\in [n] \times [m]}$.\footnote{For convenience, we label two sides as firms and workers. In practice, they may correspond to doctors and hospitals, students and colleges, teachers and schools, men and women, actual firms and workers for labor markets with fixed wages (e.g., government and union jobs, see \citealp{hall2012evidence}), or any other matching setting with limited or no transfers at all.} That is, each firm $f_i$ is endowed with a strict preference relation $\succ_{f_{i}}$ over workers and being single, i.e., $W\cup\{f_i\}$. Similarly, each worker $w_j$  is endowed with a strict preference relation $\succ_{w_{j}}$ over $F \cup \{w_j\}$. Throughout, we focus on markets where all worker-firm pairs are mutually acceptable, i.e., $w_j \succ_{f_i} f_i$ and $f_i \succ_{w_j} w_j$ for all $i$, $j$.\footnote{Our main results do not rely on all agents being acceptable and can be extended.}

A \textit{matching} is a one-to-one function $\mu:F \cup W \to F \cup W$ that \begin{inparaenum}[(i)]
\item assigns to each firm $f_i$ either a worker or herself, $\mu(f_i) \in W\cup\{f_i\}$; 
\item assigns to each worker $w_j$ either a firm or himself, $\mu(w_j) \in F \cup \{w_j\}$; 
\item is of order two, i.e, $\mu(f_i)=w_j$ if and only if $\mu(w_j)=f_i$.
\end{inparaenum} 
For a given matching $\mu$, a pair $(f_i,w_j)$ is said to form a \textit{blocking pair} if firm $f_i$ and worker $w_j$ are not matched to one another at $\mu$, but prefer each other over their assigned partners, i.e., $w_j \succ_{f_i} \mu(f_i)$ and $f_i \succ_{w_j} \mu(w_j)$. A matching is \textit{stable} if it has no blocking pairs.

\subsection{Decentralized Process} \label{section_dynamics} A folk argument suggests that if an unstable outcome is realized, decentralized interactions eventually attain stability. The leading model of such a decentralized process is \cite{roth1990random}. They show that for any unstable matching, there is a finite sequence of successive matches arising from blocking pairs that generates a stable matching. Hence, the decentralized process of allowing randomly-chosen blocking pairs to match yields stability with certainty. An advantage of this process is that convergence to stability is ensured even when agents have limited sophistication and know only local information about others' preferences.\footnote{Experimental studies suggest that subjects have a limited degree of sophistication even in relatively small decentralized markets \citep*{echenique2023experimental, pais2020decentralized}.} Theoretically, these dynamics are employed to study the formation of networks \citep{jackson2002evolution}, to justify stability as a suitable identification assumption for estimating preferences in decentralized matching markets \citep{boyd2013analyzing}, to provide the decentralized foundation of incomplete-information stability \citep{chen2020learning}, and more. 

Formally, let matching $\lambda$ be unstable. For any blocking pair $(f_i,w_j)$ of $\lambda$, a matching  $\mu$ is \textit{obtained from} $\lambda$ \textit{by satisfying} $(f_i,w_j)$ if
\begin{inparaenum}[(i)]
\item $(f_i,w_j)$ are matched to each other in $\mu$;
\item their partners at $\lambda$, if any, are unmatched at $\mu$;
\item and all other agents are matched identically under both $\mu$ and $\lambda$. 
\end{inparaenum}
In other words, we match agents inside the blocking pair, divorce their original partners, and maintain matches of all other agents. We say that there is a \textit{path} (of \textit{length} $k-1$) from matching $\lambda$ to matching $\mu$ if there exists a sequence of matchings $\lambda_1, \lambda_2, \ldots, \lambda_k$, such that $\lambda=\lambda_{1}$, $\mu=\lambda_k$, and for each $i<k$, $\lambda_{i+1}$ is obtained from $\lambda_i$ by satisfying a blocking pair. In that case, we say matching $\lambda$ can \textit{reach} or \textit{attain} matching $\mu$.\footnote{This process originates from \cite{knuth1976marriages}. Knuth assumes an equal number of firms and workers and requires divorced agents to match when satisfying a blocking pair. He showed that such a process may cycle and raised a question if we can always find a path to stability. More than a decade later, \cite{roth1990random} provided a positive answer for the setting used in the current paper. In Knuth's setting, though, sometimes it is impossible to attain stability \citep{tamura1993transformation, tan1995divorce}.}

\begin{namedthm*}{Theorem {\normalfont \citep{roth1990random}}}
\label{theorem_roth_vande_vate}
For any unstable matching, there is a finite sequence of blocking pairs that leads to a stable matching.
\end{namedthm*}

Consider a random decentralized process that starts at matching $\lambda=\lambda_1$ and successively produces matchings $\lambda_i$ in such a way that for any $i>1$, matching $\lambda_{i+1}$ is obtained from $\lambda_i$ by satisfying a single blocking pair, chosen randomly from all possible blocking pairs for $\lambda_i$. Additionally, suppose that at each step, for any two blocking pairs, one blocking pair is at most $\kappa \geq 1$ times as likely to be satisfied than the other.\footnote{Equivalently, for any unstable matching, whenever it arises, each of its blocking pairs is satisfied with a probability bounded away from zero as $i$ goes to infinity; see footnote 8 in \cite{roth1990random}.} We refer to such a random process as \textit{$\kappa$-random dynamics} for brevity. Then,

\begin{namedthm*}{Corollary {\normalfont \citep{roth1990random}}}
\label{corollary_roth_vande_vate}
For any unstable matching, the sequence of blocking pairs, produced by $\kappa$-random dynamics, converges to a stable matching with certainty.
\end{namedthm*}

The example below illustrates the dynamics of blocking-pair formation in a simple market.

\begin{example}
\label{motivating_ex_1}
Consider a market with two firms and two workers. The following payoff matrix describes the agents' ordinal preferences in a convenient way:
{\setlength\abovedisplayskip{5pt}
\setlength\belowdisplayskip{10pt}\begin{equation*}
\bordermatrix{~ &  w_1 &  w_2\cr
               f_1 &  2,1 & 1,2 \cr
               f_2 & 1,2 & 2,1\cr}.
\end{equation*}}

\noindent In this matrix notation, \begin{inparaenum}[(i)]
\item each row $i$ corresponds to firm $f_i$ and each column $j$ corresponds to worker $w_j$;
\item every $(i,j)$-entry specifies payoffs for $f_i$ and $w_j$, respectively, if they are matched with each other;
\item larger payoffs correspond to more preferred partners;
\item and all payoffs from being unmatched are normalized to zero. For instance, $f_2$ finds all workers acceptable and prefers $w_2$ to $w_1$.
\end{inparaenum}

There are two stable matchings. Matching $\mu_F$ in which agents with identical indices are matched, $\mu_F(f_i)=w_i$, is the firm-optimal stable matching: each firm weakly prefers her assigned partner to the partner she would get in any other stable matching. The worker-optimal stable matching $\mu_W$ pairs agents with different indices, $\mu_W(f_i)=w_j$ for $i = 3 - j$.\footnote{See Appendix~\hyperref[appendix-a]{A} for further details on the structure of stable matchings.}

It is easy to verify that starting from any unstable matching, the decentralized process can attain any stable matching. For instance, consider unstable matching $\lambda$ obtained from the firm-optimal stable matching $\mu_F$ by unmatching firm $f_2$ with her stable partner, worker $w_2$; that is, $\lambda(f_1)=w_1$ and $\lambda(f_2)=f_2$. Then, there exists a path, of length two, from this ``almost firm-optimal stable'' matching $\lambda$ to the worker-optimal stable matching $\mu_W$:
{\setlength\abovedisplayskip{10pt}
\setlength\belowdisplayskip{5pt}\begin{equation*}
        \lambda=\lambda_1 \xrightarrow[(f_2,w_1)]{} \lambda_2 = (\mkern-4mu\stackon{{\color{white},}f_2,}{\stackon{\mid}{w_1}}\, \stackon{{\color{white},}w_2{\color{white}f\mkern-5mu}}{\stackon{\mid}{w_2}}\mkern-5mu) \xrightarrow[(f_1,w_2)]{} \lambda_3 =\mu_W,
    \end{equation*}}
    
\noindent where under each transition arrow, we specify the corresponding blocking pair to be satisfied.

Suppose that at each step, a blocking pair is selected uniformly at random; this corresponds to $1$-random dynamics. Then, $\lambda$ converges to $\mu_W$  with probability
{\setlength\abovedisplayskip{5pt}
\setlength\belowdisplayskip{8pt}\begin{equation*}
     p \; = \underbrace{1/2}_{(f_2,w_2)} \times \;\; 0 \; + \underbrace{1/2}_{(f_2,w_1)} \times \;\; (1-p) \quad \Longrightarrow \quad p = 1/3.
\end{equation*}}

\noindent Indeed, there are two equally-likely blocking pairs for matching $\lambda$, $(f_2,w_2)$ and $(f_2,w_1)$. By satisfying pair $(f_2,w_2)$, we end up in the firm-optimal stable matching $\mu_F$; see the first term. If we satisfy blocking pair $(f_2,w_1)$, we obtain ``almost worker-optimal stable'' matching $\lambda_2$. It converges to $\mu_W$ with probability $1-p$ by symmetry; see the second term.\demo
\end{example}

In our paper, we also explore a variation of the \cite{roth1990random} dynamics, wherein agents successively select their \textit{best} blocking partners. Specifically, a blocking pair is the \textit{best blocking pair} for an agent if she prefers her partner in this pair to any other partner with whom she could form a blocking pair. Analogously to $\kappa$-random dynamics, we define \textit{$\kappa$-random best dynamics} by substituting ``blocking pairs'' with ``\textit{best} blocking pairs'' in the definition of $\kappa$-random dynamics. In fact, due to \cite{ackermann2011uncoordinated}, for any unstable matching, there exists a finite sequence of best blocking pairs that leads to a stable matching. Consequently, results similar to the above-mentioned theorem and corollary hold when we restrict attention to best blocking pairs.

Both versions of random dynamics particularly allow for non-uniform and time-dependent formation of (best) blocking pairs. The probability that a specific blocking pair of agents matches at a given time can be quite general, reflecting factors such as the likelihood that these agents would encounter each other \citep{roth1990random}. Alternatively, the probability may depend on the agents' incentives to form a match. One may assume, for instance, that agents have cardinal payoffs, and a blocking pair generating a higher total surplus---or a higher total surplus gain compared to the former match---is more likely to be formed. We present our main results under these broad classes of decentralized dynamics.

\section{Anything Goes}
\label{section_characterization}
In this section, we prove that, under mild conditions, starting from \textit{any} unstable matching, decentralized interactions can attain \textit{any} stable matching. For simplicity, we consider balanced markets, with an equal number $n$ of agents on both sides, for the rest of the paper.\footnote{All definitions and results can be generalized in a straightforward manner to imbalanced markets.} 

Formally, consider a subset of firms and a subset of workers of equal size $k<n$. They form a submarket with preferences inherited from the original market. Suppose that this submarket has a stable matching such that any agent---firm or worker---inside the submarket prefers her stable partner under that matching to every agent outside the submarket. In this case, we say that the above two subsets constitute a \textit{fragment}. A fragment is \textit{trivial} if all stable matchings in the original market agree on how they match agents inside the fragment.

In addition, we say that a matching is \textit{almost stable} if it has a path of length one to a stable matching; that is, it is just one blocking pair away from stability.\footnote{This concept is partially related to \textit{almost stable matchings} of \cite*{abraham2005almost} and \textit{minimally unstable matchings} of \cite{dougan2021minimally}. These papers consider settings in which a stable matching may not exist and focus on matchings that are ``as stable as possible.''} Notably, the theorem below holds even when we restrict attention to best blocking pairs.

\begin{theorem}
\label{theorem_characterization}
The following three statements are equivalent:
\begin{enumerate}[(i)]
    \item for any unstable matching $\lambda$ and any stable matching $\mu$, there exists a finite sequence of (best) blocking pairs that leads from $\lambda$ to $\mu$;
    \item for any almost stable matching $\lambda$ and any stable matching $\mu$, there exists a finite sequence of (best) blocking pairs that leads from $\lambda$ to $\mu$;
    \item there are no non-trivial fragments.
\end{enumerate}
\end{theorem}

This result shows one type of fragility of stable matchings. Small perturbations of a stable matching may yield \textit{any} stable matching, close to or distant from the perturbed matching. 

\begin{corollary}
\label{corollary_characterization}
Suppose that there are no non-trivial fragments. For any unstable matching, the sequence of (best) blocking pairs, produced by $\kappa$-random (best) dynamics, converges to any stable matching with positive probability.
\end{corollary}

This corollary is further refined in Sections~\ref{section_exponential} and~\ref{section_experiments}, illustrating that decentralized interactions often lead the market away from a slightly perturbed stable matching, with no guarantee of returning or even converging to a close, similar stable matching.

In what follows, we discuss fragments in more detail and elucidate their role in constraining decentralized interactions. We argue that the absence of non-trivial fragments is a mild condition. We then present a sketch of the proof of the theorem.

\subsection{Fragments}
\label{subsection_fragments}
Consider subsets $\bar{F} \subsetneq F$ and $\bar{W} \subsetneq W$ of equal size $k<n$. Let $\bar{\mathcal{M}}$ be the market induced by the original market  $\mathcal{M}$ when restricted to $\bar{F} \times \bar{W}$. Firms $\bar{F}$ and workers $\bar{W}$ form a \textit{fragment}, of \textit{size} $k$, in $\mathcal{M}$ if there exists a stable matching $\bar{\mu}: \bar{F} \cup \bar{W} \to \bar{F} \cup \bar{W}$ for $\bar{\mathcal{M}}$ such that 
\begin{inparaenum}[(i)]
\item for any $\bar{f} \in \bar{F}$ and $w \notin \bar{W}$, $\bar{\mu}(\bar{f}) \succ_{\bar{f}} w$, 
\item and for any $\bar{w} \in \bar{W}$ and $f \notin \bar{F}$, $\bar{\mu}(\bar{w}) \succ_{\bar{w}} f$.
\end{inparaenum}
In this case, we say that matching $\bar{\mu}$ \textit{induces} fragment $(\bar{F}, \bar{W})$.\footnote{A market may have multiple fragments: nested, overlapping, or disjoint. In addition, a fragment might be induced by multiple matchings; see Example 3 in the Online Appendix. For any inducing matching, we can find a stable matching in the original market that agrees with that inducing matching over its corresponding fragment, as shown by Lemma 4 in the Online Appendix. The converse is not true. That is, a stable matching in the original market might disagree with all matchings inducing a fragment---in Example~\ref{motivating_ex_2} below, the worker-optimal stable matching $\mu_W$ disagrees with the unique inducing matching $\bar{\mu}$.} 

Fragments have the following defining property: any stable matching in the original market must match agents inside a fragment within the fragment. That is why we call them ``fragments.'' Nonetheless, stable matchings may still disagree with each other, and with all inducing matchings, on how they match agents inside the fragment. 

\begin{lemma}
\label{lemma_definition}
Suppose that firms  $\bar{F}$ and workers $\bar{W}$ form a fragment. Then, $\mu(\bar{F})=\bar{W}$ for for any stable matching $\mu$ in the original market.
\end{lemma}

Intuitively, consider the fragment induced by matching $\bar{\mu}(f_i)=w_i$, $i \leq k < n$. Assume by contradiction that, say, worker $w_1$ is not matched within the fragment for some stable matching $\mu$. Then, he prefers his partner under $\bar{\mu}$, firm $f_1$. By the stability of $\mu$, firm $f_1$ prefers her partner under $\mu$ to worker $w_1$. Therefore, by the definition of a fragment, $f_1$ must be matched within the fragment under $\mu$, say, with worker $w_2$. By the stability of $\bar{\mu}$ for the fragment, $w_2$ prefers his stable partner under $\bar{\mu}$, firm $f_2$, to firm $f_1$. This implies that $f_2$ prefers her partner under $\mu$ to worker $w_2$, again by the stability of $\mu$. Hence, $f_2$ also must be matched within the fragment under $\mu$, say, with worker $w_3$. We can proceed iteratively until we reach the last firm, firm $f_k$, that cannot be matched within the fragment since all workers inside the fragment are already matched with other firms. This yields a contradiction.

A fragment is called \textit{trivial} if all stable matchings in the original market agree on how they match agents within the fragment; this is more restrictive than the conclusion of Lemma~\ref{lemma_definition}. If a fragment is trivial, it has the unique inducing matching, and all stable matchings in the original market coincide with that inducing matching on the fragment. For instance, a firm-worker pair that are each other's favorite partner constitutes a trivial fragment of size one. This pair is called a \textit{top-top match}. Furthermore, a sequence of $k < n$ such pairs---where each new top-top pair in the sequence is obtained from the submarket after removing all previous top-top pairs from the original market---forms a trivial fragment of size $k$. Nonetheless, trivial fragments are not limited to such sequences, see Example 4 in the Online Appendix. The following example shows how a \textit{non-trivial} fragment might restrain the dynamics.

\begin{example}
\label{motivating_ex_2}
Consider a market with three firms and three workers
{\setlength\abovedisplayskip{5pt}
\setlength\belowdisplayskip{10pt}
\begin{equation*}
\bordermatrix{~ & \color{orange} w_1 &  \color{orange} w_2 & w_3\cr
 \color{orange} f_1 &  \color{orange} \bf 3,2 &  \color{orange} 1,3 & 2,1  \cr
 \color{orange} f_2 &  \color{orange} 1,3 &  \color{orange} \bf 3,2 & 2,2 \cr
f_3 & 3,1 & 2,1 & 1,3 \cr}
\end{equation*}}

\noindent that has two stable matchings, $\mu_F = (f_1,f_2,f_3)$ and $\mu_W = (f_2, f_1, f_3)$. As a side note, for markets with more than two workers, we use a compact matching notation that specifies partners of workers $w_1$, $w_2$, and so forth.

Firms $\bar{F}=\{f_1,f_2\}$ and workers $\bar{W}=\{w_1,w_2\}$ form a fragment induced by matching $\bar{\mu}$ that couples firms $f_1$ and $f_2$ with workers $w_1$ and $w_2$, respectively. It is non-trivial since the worker-optimal stable matching $\mu_W$ disagrees with the inducing matching $\bar{\mu}$ on the fragment.

In contrast to Example~\ref{motivating_ex_1}, some unstable matchings cannot attain all stable ones. Consider almost stable matching $\lambda=(f_1,f_2,w_3)$ that agrees with the inducing matching $\bar{\mu}$ on the fragment. Matching $\bar{\mu}$ traps the decentralized process. Indeed, since $\bar{\mu}$ is stable for the fragment, no agent inside it can form a blocking pair with anyone that is also inside. Also, by definition, every agent inside the fragment prefers her partner under $\bar{\mu}$ to every agent outside it. Altogether, no agent inside the fragment can form a blocking pair with anyone, neither inside nor outside it. Thus, the dynamics cannot reach any stable matching that disagrees with $\bar{\mu}$. Non-triviality ensures that such a stable matching---here, $\mu_W$---exists.\demo

\end{example}

We argue that a non-trivial fragment is demanding and most markets have no such fragments. A non-trivial fragment requires that \begin{inparaenum}[(i)]
\item agents inside the fragment are matched in a \textit{stable} way; 
\item they prefer their stable partners to \textit{everyone} outside;
\item and there is \textit{another stable} way to match these agents in the original market.
\end{inparaenum}
As an illustration, consider markets with uniformly random preferences. Numerical simulations show that, with high probability, there are no non-trivial fragments at all; we provide theoretical results regarding probabilistic aspects of fragments in the Online Appendix.\footnote{Fragments are a novel concept, and their analysis appears rather delicate. The Online Appendix presents results on the likelihood of fragments and discusses associated challenges. Specifically, we demonstrate that, with high probability, there are no large fragments at all, including trivial ones. We also introduce techniques, in the spirit of \cite{pittel1989average}, that may be potentially useful for future research.} Figure~\ref{figure_no_fragments} illustrates that, although the frequency of markets with non-trivial fragments is increasing for market sizes of up to $n=7$, this frequency never exceeds 20\% and eventually decreases with the market size.\footnote{We simulate $n\times n$ markets, the number of which varies from 1,000 for larger $n$ to 20,000 for smaller $n$.} Thus, for such markets, Theorem~\ref{theorem_characterization} holds almost generically.\footnote{At the same time, balanced markets with uniformly random preferences have many stable matchings---all of which are anticipated to be fragile---with significantly different welfare properties. Formally, \cite{pittel1989average} shows that, asymptotically, the expected number of stable matchings is $e^{-1}n\ln n$. Moreover, the firms' average ranks of their assigned workers under the firm-optimal ($\ln n$) and worker-optimal ($n/\ln n$) stable matchings are substantially different; analogous results hold for workers. The above exercise becomes less relevant for imbalanced markets with uniformly random preferences \citep*{ashlagi2017unbalanced}.}

\begin{figure}[h]
    \centering
    \includegraphics[width=0.5\linewidth]{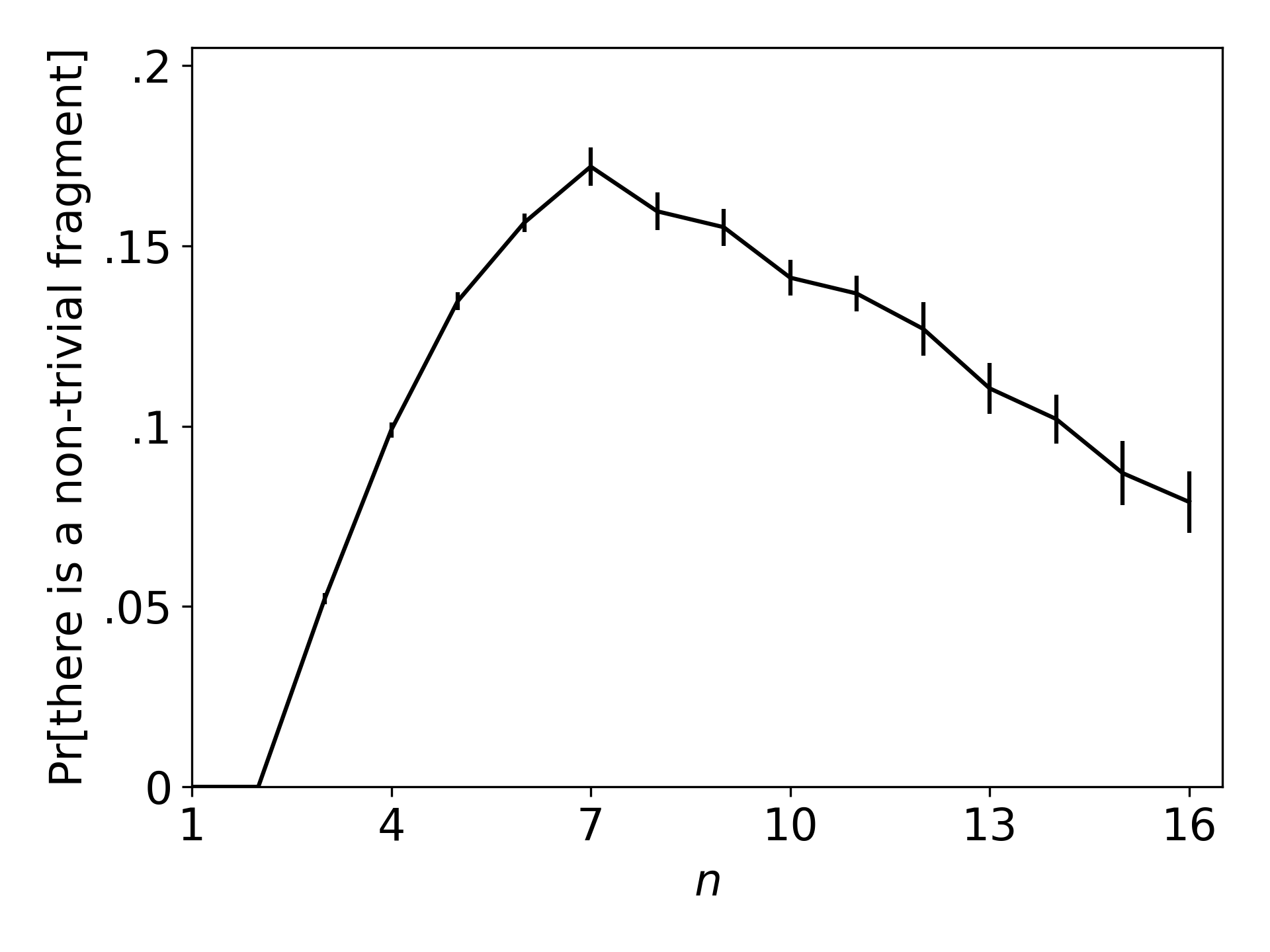}
    \captionsetup{justification=centering}
    \caption{Frequency of random markets with non-trivial fragments}
    \label{figure_no_fragments}
\end{figure}

Furthermore, even when non-trivial fragments are present, they typically restrain the dynamics only for a small fraction of specific unstable matchings. The great majority of unstable matchings can still attain many, if not all, stable matchings. 

\subsection{Sketch of the Proof}

In the proof of the theorem, we use one more piece of notation. Let $\mu_{-f}$ denote the almost stable matching obtained from stable matching $\mu$ by unmatching firm $f$ with her stable partner, worker $\mu(f)$.

We prove the theorem by establishing the following cycle of implications: $(ii) \Rightarrow (i) \Rightarrow (iii) \Rightarrow (ii)$. To avoid technicalities, we provide a sketch of the proof for the case when no restriction is placed on blocking pairs; the full proof is relegated to Appendix~\hyperref[appendix-b]{B}. The first implication $(ii) \Rightarrow (i)$ follows from the stabilization theorem of \cite{roth1990random}. Specifically, consider any unstable matching $\lambda$. By the stabilization result, there is a finite sequence of blocking pairs that leads to some stable matching $\mu$:
{\setlength\abovedisplayskip{5pt}
\setlength\belowdisplayskip{5pt}
\begin{equation*}
    \lambda=\lambda_1 \longrightarrow \lambda_2 \longrightarrow  \ldots \longrightarrow  {\color{orange}\lambda_{k-1}} \longrightarrow  \mu = \lambda_{k}.
\end{equation*}}Therefore, $\lambda$ can reach almost stable matching $\lambda_{k-1}$, from which in turn it can attain any stable matching by the premise $(ii)$. 

The second implication $(i) \Rightarrow (iii)$ holds by the definition of a non-trivial fragment. Suppose by contradiction that there is a non-trivial fragment. Consider an unstable matching that, when restricted to the fragment, agrees with its inducing matching. Then, this unstable matching cannot yield any stable matching that instead disagrees with the inducing matching over the fragment. Since the fragment is non-trivial, such a stable matching exists, thus contradicting the premise $(i)$; this argument is similar to Example~\ref{motivating_ex_2}. 

Finally, and most importantly, why can decentralized interactions lead to any stable matching if there are no non-trivial fragments? The third implication $(iii) \Rightarrow (ii)$ is the key part of the proof. We prove this implication by induction on the size $n$ of matching markets with no non-trivial fragments. The result is trivial for $n=1$, when there is only one agent on each side of the market. It also holds for such markets of size $n=2$. This is because it is valid for the market from Example~\ref{motivating_ex_1} that is the only $2\times 2$ market with multiple stable matchings, up to permutations of firms' and workers' indices.

To glean intuition behind the induction step, assume that the result holds for all markets with no non-trivial fragments of size up to $n-1$ and focus on any such market of size $n$. Consider any almost stable matching $\lambda$ which corresponds to some stable matching $\mu$. The proof consists of two main observations. The first observation is at the core of the proof and crucially relies on the properties of fragments as well as the induction hypothesis. We show that $\lambda$ can reach either any stable matching---as desired---or any other almost stable matching for stable matching $\mu$. Hence, suppose that any almost stable matching for $\mu$ can be attained. The second observation, which we also describe in more detail later in this section, is more technical. This observation connects our procedure to the deferred acceptance algorithm in order to construct a path from one of those almost stable matchings to a new almost stable matching corresponding to some new stable matching. We employ these two observations repeatedly until we find paths to all stable matchings.

To illustrate the first observation, consider any market with no non-trivial fragments of size $n=3$ and an arbitrary almost stable matching $\lambda$. Up to relabeling, almost stable $\lambda=\mu_{-f_3}$ is obtained from stable matching $\mu=(f_1,f_2,f_3)$ by unmatching firm $f_3$ with her stable partner, worker $w_3$. We show that $\mu_{-f_3}$ can attain either any stable matching or any other almost stable matching $\mu_{-f_i}$, $i < 3$.

Focus on the subsets $\bar{F}=\{f_1,f_2\}$ and $\bar{W}=\{w_1,w_2\}$ of firms and workers, respectively. Suppose first that these subsets form a fragment. This fragment must be trivial by the premise. Therefore, all stable matchings match $f_1$ with $w_1$ and $f_2$ with $w_2$. This implies that matching $\mu$ is the unique stable matching, so the result follows immediately, say, by the stabilization theorem. Henceforth, suppose that the subsets $\bar{F}$ and $\bar{W}$ do not constitute a fragment. Since stable matching $\mu$ remains stable when restricted to the submarket induced by these subsets, there is an agent inside these subsets that prefers someone outside to her stable partner under $\mu$. By symmetry, let firm $f_2$ prefer worker $w_3$ to her stable partner, worker $w_2$. By the stability of matching $\mu$ in the original market, $w_3$ must prefer his stable partner, firm $f_3$, to firm $f_2$. But then almost stable matching $\mu_{-f_2}$ can be attained by successively satisfying blocking pairs $(f_2,w_3)$ and $(f_3,w_3)$. This can be described by the following diagram in which solid circles represent initial matches:
\begin{figure}[htb]
    \centering
\begin{tikzpicture}
 \tikzstyle{main node} = [circle,draw,minimum size=0.4cm,inner sep=0pt, text centered, node distance=1cm]
  \tikzstyle{label node} = [circle,draw = white,minimum size=0.4cm,inner sep=0pt, align=right, node distance=1cm]
    \node[main node, label=above:{$w_1$}, label=left:{$f_1$}] (a1) {};
    \node[main node, color = white, right = of a1, label=above:{$w_2$}] (a2) {};
    \node[main node, color = white, right = of a2, label=above:{$w_3$}] (a3) {};
    \node[main node, color = white, below = of a1, label=left:{$f_2$}] (b1) {};
    \node[main node, color = orange, below = of a2] (b2) {};
    \node[main node, color = orange, dashed, below = of a3] (b3) {};
    \node[main node, color = white, below = of b1, label=left:{$f_3$}] (c1) {};
    \node[main node, color = white, below = of b2] (c2) {};
    \node[main node, color = orange, dashed, below = of b3] (c3) {};
    
    %\path (b2) -- node[label=above:{Step 1}] (pair_2) {} (b3);
    \draw[-{Latex[width=2mm]}, color = orange] (b2) -- (b3);
    %\path (b3) -- node[label={[rotate=270,text depth=1ex]Step 2}] (pair_2) {} (c3);
    \draw[-{Latex[width=2mm]}, color = orange] (b3) -- (c3);
\end{tikzpicture}
\end{figure}

\noindent Thus, $\mu_{-f_3}$ can attain the new almost stable matching $\mu_{-f_2}$.

We next consider smaller subsets of firms $\bar{F}=\{f_1\}$ and workers $\bar{W}=\{w_1\}$. Suppose first that they form a fragment. It must be trivial by the premise; in fact, top-top match $(f_1,w_1)$ is a trivial fragment by definition. Thus, all stable matchings match $f_1$ with $w_1$. As it turns out, after removing agents forming a trivial fragment from the original market, the remaining market---that consists of firms $F\setminus \bar{F}$ and workers $W\setminus\bar{W}$---continues to have no non-trivial fragments (Lemma~\ref{lemma_submarket} in Appendix~\hyperref[appendix-b]{B}). When restricted to the remaining market, $\mu_{-f_3}$ is unstable and, by the induction hypothesis, can reach any stable matching in this smaller market. Since all stable matchings in the original market agree on the removed trivial fragment and persist to be stable in the remaining market, matching $\mu_{-f_3}$ can reach all stable matchings in the original market, as desired. Henceforth, assume that $\bar{F}$ and $\bar{W}$ do not form a fragment. As before, by symmetry, let firm $f_1$ prefer some worker outside, $w_2$ or $w_3$, to her stable partner, worker $w_1$. In either of these two cases, the stability of matching $\mu$ ensures that matching $\mu_{-f_1}$ can be attained in two steps from either $\mu_{-f_2}$ or $\mu_{-f_3}$:

\begin{figure}[htb]
    \centering
\begin{tikzpicture}
 \tikzstyle{main node} = [circle,draw,minimum size=0.4cm,inner sep=0pt, text centered, node distance=1cm]
  \tikzstyle{label node} = [circle,draw = white,minimum size=0.4cm,inner sep=0pt, align=right, node distance=1cm]
    \node[main node, color = orange, label=above:{$w_1$}, label=left:{$f_1$}] (a1) {};
    \node[main node, color = orange, dashed, right = of a1, label=above:{$w_2$}] (a2) {};
    \node[main node, color = white, right = of a2, label=above:{$w_3$}] (a3) {};
    \node[main node, color = white, below = of a1, label=left:{$f_2$}] (b1) {};
    \node[main node, color = orange, dashed, below = of a2] (b2) {};
    \node[main node, color = white, below = of a3] (b3) {};
    \node[main node, color = white, below = of b1, label=left:{$f_3$}] (c1) {};
    \node[main node, color = white, below = of b2] (c2) {};
    \node[main node, color = black, below = of b3] (c3) {};
    
    %\path (a1) -- node[label=below:{Step 1}] (pair_2) {} (a2);
    \draw[-{Latex[width=2mm]}, color = orange] (a1) -- (a2);
    %\path (a2) -- node[label={[rotate=270,text depth=1ex]Step 2}] (pair_2) {} (b2);
    \draw[-{Latex[width=2mm]}, color = orange] (a2) -- (b2);
\end{tikzpicture}
\qquad \qquad
\begin{tikzpicture}
 \tikzstyle{main node} = [circle,draw,minimum size=0.4cm,inner sep=0pt, text centered, node distance=1cm]
  \tikzstyle{label node} = [circle,draw = white,minimum size=0.4cm,inner sep=0pt, align=right, node distance=1cm]
    \node[main node, color = orange, label=above:{$w_1$}, label=left:{$f_1$}] (a1) {};
    \node[main node, color = white, right = of a1, label=above:{$w_2$}] (a2) {};
    \node[main node, color = orange, dashed, right = of a2, label=above:{$w_3$}] (a3) {};
    \node[main node, color = white, below = of a1, label=left:{$f_2$}] (b1) {};
    \node[main node, color = black, below = of a2] (b2) {};
    \node[main node, color = white, dashed, below = of a3] (b3) {};
    \node[main node, color = white, below = of b1, label=left:{$f_3$}] (c1) {};
    \node[main node, color = white, below = of b2] (c2) {};
    \node[main node, color = orange, dashed, below = of b3] (c3) {};
    %\path (b2) -- node[label=above:{Step 1}] (pair_2) {} (b3);
    \draw[-{Latex[width=2mm]}, color = orange] (a1) -- (a3);
    %\path (b3) -- node[label={[rotate=270,text depth=1ex]Step 2}] (pair_2) {} (c3);
    \draw[-{Latex[width=2mm]}, color = orange] (a3) -- (c3);
\end{tikzpicture}
\end{figure}
\noindent Consequently, $\mu_{-f_3}$ can reach any other almost stable matching $\mu_{-f_i}$, $i < 3$, corresponding to $\mu$. This concludes the first observation.

In the second observation, we connect the decentralized process of our interest to the deferred acceptance algorithm. \cite{mcvitie1971stable} adapt this algorithm by introducing a new operation to determine the set of all stable matchings. Their adapted version can be roughly described as follows; see the full details and relevant results in Appendix~\hyperref[appendix-a]{A}. They first run the standard firm-proposing version to compute the firm-optimal stable matching $\mu_F$, the best stable matching for firms. Next, they seek to obtain any other stable matching. Particularly, fix an arbitrary stable matching $\mu \neq \mu_F$. Then, some firm $f$ prefers worker $w=\mu_F(f)$ to worker $\mu(f)$. \cite{mcvitie1971stable}'s idea is to break the $(f,w)$-partnership and restart the previously terminated deferred acceptance algorithm by forcing firm $f$ to propose to the worker following $w$ in her list. As it turns out, this operation called \textit{breakmarriage} generates a new stable matching $\mu'$ that is at least as good as matching $\mu$ for all firms. If $\mu$ and $\mu'$ are different, some firm $f'$ prefers worker $w'=\mu'(f')$ to worker  $\mu(f')$. \cite{mcvitie1971stable} then break the $(f',w')$-partnership, restart the deferred acceptance algorithm one more time, and proceed iteratively. They show that any stable matching can be obtained by successive applications of breakmarriage operations. In Lemma~\ref{lemma_breakmarriage} in Appendix~\hyperref[appendix-a]{A}, we prove that our dynamics can replicate the restarted deferred acceptance algorithm following breakups for all relevant breakmarriage operations.

We conclude the proof by combining the two observations. The first observation suggests that the absence of non-trivial fragments effectively allows to attain any almost stable matching for a particular stable matching. Equivalently, we can break any stable partnership in that stable matching. This, together with the second observation, implies that all relevant breakmarriage operations in \cite{mcvitie1971stable}'s algorithm can be emulated by the dynamics. We use this connection to establish the result.

\section{Exponential Stabilization}
\label{section_exponential}
We now turn to another form of fragility. Even though convergence to stability is always guaranteed, the market may be far from stable for long periods of time. To isolate this form of fragility, this section examines markets with a unique stable matching. In such markets, the dynamics necessarily converge to one matching, the unique stable one. As mentioned earlier, our focus on markets with a unique stable matching is also in line with the literature suggesting that large markets entail small cores.\footnote{Specifically, theoretical work has identified various conditions under which in large markets, all agents, except for a vanishing proportion, are matched with the same partners in all stable matchings; see \cite*{Immorlica2005}, \cite*{kojima2009incentives}, and \cite*{ashlagi2017unbalanced}, as well as our discussion in footnote~\ref{footnote_mixed_evidence}. Although our first theorem applies to such markets and suggests re-equilibration dynamics, it becomes less relevant since all stable matchings are essentially identical.} Henceforth, we analyze the random dynamics introduced in Section~\ref{section_dynamics}, which allow for non-uniform and time-dependent sequential formation of (best) blocking pairs.

Our second theorem shows that for a large class of markets with a unique stable matching, even a small deviation from stability leads to an exponentially long stabilization process. More precisely, we consider \textit{any} market with a unique stable matching. We prove that the market can be augmented by adding a small fraction of new agents so that the resulting market still has a unique stable matching, but where even for small perturbations of the stable matching, stabilization takes an exponentially long time.\footnote{Our approach is similar in spirit to the ``maximal domain'' exercises common in the matching literature (e.g., see \citealp*{gul1999walrasian, hatfield2012matching,kamada2023fair}) and to the augmentation exercises in \cite*{fernandez2022centralized}.}

Formally, a \textit{$\delta$-augmented market} of the original market of size $n$ is a new market of size at most $(1+\delta)n$ obtained from the original market by adding at most a $\delta$-proportion of new agents on both sides, so that \begin{inparaenum}[(i)]
\item both the original and the augmented markets coincide when restricted to the original agents;
\item the augmented market has a unique stable matching;
\item and all original agents have identical stable partners in both markets.
\end{inparaenum} 
Importantly, an augmentation does not alter the original market and its structure of stable matchings.

We are interested in small deviations from stability. A matching is \textit{$\epsilon$-unstable} if at least an $\epsilon$-fraction of firms, and thus workers, are not matched with their unique stable partners.

\begin{theorem}
\label{theorem_exponential}
Focus on $\kappa$-random (best) dynamics. Consider any sequence of markets of size $n \in \mathbb{N}$ with a unique stable matching and any $\delta>0$,  $\epsilon \in (0,1]$, and $\kappa \geq 1$. Then, there exists a corresponding sequence of $\delta$-augmented markets for which any sequence of $\epsilon$-unstable matchings with probability $1-2^{-\Omega(n)}$ needs $2^{\Omega(n)}$ steps to regain stability.\footnote{We write $f(n) = \Omega(g(n))$ if $g(n)=O(f(n))$.}
\end{theorem}

This theorem emphasizes an additional fragility aspect of stable matchings; even for slight deviations from stability, the stabilization process takes an extremely long time. A sketch of the proof of the theorem is provided in Section~\ref{section_proof_exponential}. At a very high level, we construct augmented markets in which, for many matchings along any path to stability, there are substantially more ``destabilizing'' (best) blocking pairs that move the market away from stability than ``stabilizing'' (best) blocking pairs---corresponding to stable pairs of agents---that instead facilitate the convergence. This significantly slows down the convergence and causes exponentially long stabilization paths.\footnote{In our analysis of Theorem~\ref{theorem_exponential}, we focus on random dynamics that guarantee convergence to stability, consistent with our discussion in Section~\ref{section_dynamics}. To produce long paths---though not necessarily leading to stability---an even broader class of dynamics suffices, wherein ``destabilizing'' pairs are not vastly less likely to be satisfied compared to other pairs.  Moreover, it seems feasible to allow $\kappa=\kappa(n)$ to increase with $n$, albeit at the expense of introducing unnecessary complexity into the analysis.}

Theorem~\ref{theorem_exponential}, along with its underlying intuition, naturally complements Theorem~\ref{theorem_characterization}. In fact, the simulations in Section~\ref{section_experiments} demonstrate that for random markets, irrespective of whether stable matchings are unique, and even when initializing a market at almost stable matchings, the stabilization process typically leads the market away from stability and takes a very long time to regain it.\footnote{Obtaining this result theoretically is challenging due to two layers of randomness inherent in both random markets and random decentralized dynamics. Addressing this challenge would require not only alternative methods but also specific assumptions about randomness. Instead, similar to our first theorem, Theorem~\ref{theorem_exponential} applies to non-random markets, specifically to all markets having a unique stable matching. Moreover, our result holds for a broad range of (potentially non-uniform and time-dependent) random dynamics.} Due to the prevalence of ``destabilizing'' rematching opportunities, many market participants are mismatched for extended periods. Furthermore, for markets with multiple stable matchings, even almost stable matchings converge with substantial probability to a stable matching distant from the minimally perturbed matching; this further refines our findings in Section~\ref{section_characterization}, particularly Corollary~\ref{corollary_characterization}.

In the rest of this section, we first outline the proof of the second theorem. We then discuss one more way in which fragments can restrain decentralized interactions.

\subsection{Sketch of the Proof} \label{section_proof_exponential} 
The proof proceeds in two steps. First, we establish an analogous result for a class of non-augmented markets having a unique stable matching. Then, we utilize markets from this class to augment arbitrary markets and demonstrate the theorem.

For the first step, consider a market of size $n\in \mathbb{N}$ with a unique stable matching $\mu$, where for some $\eta \in (0,1)$, the following two conditions hold: 
\begin{inparaenum}[(i)]
\item any firm \textit{except} firm $\bar{f}$ is preferred by at least an $\eta$-proportion of workers to their stable partners;
\item and similarly, any worker \textit{except} worker $\bar{w}$ is preferred by at least an $\eta$-proportion of firms to their stable partners.
\end{inparaenum}
These exceptions arise because in a market with a unique stable matching, it is impossible for \textit{every} agent on one side to be preferred by some agents from the other side to their stable partners.\footnote{See Lemma 7 in \cite{balinski1997stable}.} We prove that any sequence of such markets (indexed by sizes $n\in \mathbb{N}$), where $\eta$ is fixed along the sequence, achieves a conclusion parallel to that of the theorem; see Proposition~\ref{proposition_class} in Appendix~\hyperref[appendix-c]{C}. The identified class of markets, which is non-empty and quite broad, particularly includes certain ``almost'' fully assortative markets, wherein all agents on each side agree on how they rank the vast majority of all agents from the other side.

To illustrate the idea behind Proposition~\ref{proposition_class}, consider the market of size $n$ described earlier. For any matching $\lambda$, let $\mathcal{S}(\lambda)$ denote the number of firms, and hence workers, matched with their stable partners. In order to restore stability, $\mathcal{S}(\mu)=n$, at some point of the process, the market must enter the region $\mathcal{S}(\lambda) \geq (1-\zeta)n$, where $\zeta \leq \epsilon$, and remain in this region until convergence; in fact, $\zeta>0$ can be arbitrarily small. Therefore, consider any unstable matching $\lambda$ with $\mathcal{S}(\lambda) \geq (1-\zeta)n$. To avoid technical difficulties, suppose $\lambda$ satisfies the following condition {\color{orange}($\star$)}: some agent $a \notin\{\bar{f}, \bar{w}\}$ is unmatched. In that case, this agent is preferred by at least an $\eta$-proportion of agents from the other side to their stable partners. Since at least a $(1-\zeta)$-fraction of agents are matched with their stable partners in matching $\lambda$, there are at least $(\eta-\zeta)n$ destabilizing (best) blocking pairs that generate matchings $\lambda'$ with $\mathcal{S}(\lambda')=\mathcal{S}(\lambda)-1$; note, however, that new matchings $\lambda'$ might fail to satisfy {\color{orange}($\star$)}. At the same time, there are at most $\zeta n$ stabilizing (best) blocking pairs that lead to matchings $\lambda'$ with $\mathcal{S}(\lambda')=\mathcal{S}(\lambda)+1$; they may not satisfy {\color{orange}($\star$)} as well. For small enough $\zeta>0$, there are many more destabilizing pairs than stabilizing ones. If all matchings along the process satisfied {\color{orange}($\star$)}, we could associate the dynamics with a random walk that is heavily biased in the ``destabilizing'' direction. Such random walks are known to take exponentially many steps, with high probability, to reach $\mathcal{S(\mu)}=n$. In our proof, since some matchings resulting from the dynamics violate {\color{orange}($\star$)}, we use more sophisticated arguments to establish the result.

In the second step of the proof, we can use any sequence of constructed markets from the identified class, for which Proposition~\ref{proposition_class} holds, to $\delta$-augment the original markets and establish the theorem by analogous techniques. For each original market, of size $n$, add new agents with preferences over each other that mirror the constructed market of size $\floor{\delta n}$. Additionally, let any agent, original or new, prefer original agents to new ones. The resulting market, of size $\floor{\delta n}$, is indeed $\delta$-augmented and satisfies conditions akin to those described in Proposition~\ref{proposition_class}. Despite the sizes of the $\delta$-augmented markets being $n+\floor{\delta n}$ rather than $n$, they still increase linearly with $n \in \mathbb{N}$. The desired exponential convergence result then follows by similar arguments as before, completing the proof of the theorem.

\subsection{Fragments Revisited} 
Section~\ref{section_characterization}, and notably Theorem~\ref{theorem_characterization}, sheds light on the special role of non-trivial fragments in constraining the stabilization dynamics. Below, we illustrate another way in which fragments, particularly trivial ones, restrain decentralized interactions.

Specifically, consider a market in which firms and workers can be relabeled so that subsets $\bar{F}_1=\{f_1\}$ and $\bar{W}_1=\{w_1\}$ form a fragment of size one, $\bar{F}_2=\bar{F}_1 \cup \{f_2\}$ and $\bar{W}_2=\bar{W}_1 \cup \{w_2\}$ constitute a fragment of size two, and so forth. In particular, firms $\bar{F}_{n-1}=\bar{F}_{n-2} \cup \{f_{n-1}\} = \{f_1, f_2,\ldots, f_{n-1}\}$ and workers $\bar{W}_{n-1}=\bar{W}_{n-2} \cup \{w_{n-1}\} = \{w_1, w_2,\ldots, w_{n-1}\}$ generate a fragment of size $n-1$. In other words, a market has a \textit{nested structure of fragments}; all these fragments are trivial by definition. This market can also be equivalently described as a sequence of top-top match pairs. Firm-worker pair $(f_1,w_1)$ is a top-top match: they are each other's favorite partner. Furthermore, for any $i \geq 2$, pair $(f_i,w_i)$ is a top-top match in the submarket formed by firms $F \setminus \bar{F}_{i-1}$ and workers $W \setminus \bar{W}_{i-1}$. This alternative formulation coincides with the \textit{sequential preference condition}, first introduced by \cite{eeckhout2000uniqueness}. 

Any such market has a unique stable matching that can be derived by successively partnering top-top match pairs. Many matching markets studied in the literature have a nested structure of fragments.\footnote{For instance, the $\alpha$-reducibility property \citep{clark_uniqueness_2006}, the aligned preferences condition \citep*{ferdowsian2023decentralized}, and oriented preferences \citep{reny_simple_2021} satisfy \cite{eeckhout2000uniqueness}'s condition and thus generate a nested structure of fragments.} This class includes assortative markets, in which agents on one or both sides share the same ranking of agents from the other side.

A nested structure of fragments makes the stabilization dynamics polynomial. Consider an arbitrary unstable matching and, as an example, focus on $1$-random dynamics: at each step, a blocking pair is selected uniformly at random. If firm $f_1$ and worker $w_1$ from top-top match pair $(f_1,w_1)$ are not yet matched with each other, this pair remains a blocking pair until they are matched. Since there are at most $n^2$ blocking pairs at each step of the dynamics, they are expected to match in $O(n^2)$ time. Once matched, $f_1$ and $w_1$ continue to be partners until convergence, which means the decentralized process gets stuck in fragment $(\bar{F}_1, \bar{W}_1)$.  Next, if firm $f_2$ and worker $w_2$ from top-top match pair $(f_2,w_2)$ are not yet matched, this pair remains a blocking pair until it is resolved. These agents are expected to match with each other in $O(n^2)$ time as well. The dynamics then get stuck in fragment $(\bar{F}_2, \bar{W}_2)$. By iterative arguments, the stabilization takes $O(n^3)$ expected time.\footnote{\cite{lebedev2007using} and \cite{ackermann2011uncoordinated} show the polynomial convergence under uniform dynamics for markets with aligned preferences, a special case of markets with a nested structure of fragments.}

Nonetheless, even markets in which the convergence is polynomial are potentially fragile, as demonstrated by Theorem~\ref{theorem_exponential}. A slight proportion of agents can be added to make the stabilization exponential. In fact, the construction in the theorem's proof can be straightforwardly modified to ensure that augmented markets have no fragments at all. 

\section{Computational Experiments}
\label{section_experiments}

This section presents simulation results for random markets that complement and sharpen our theoretical findings. In markets with multiple stable matchings, even a minimal perturbation of a stable matching often converges to a stable matching distant from the perturbed matching. Furthermore, regardless of whether stable matchings are unique, the stabilization dynamics typically stray from stability, take a long time to reach stability, and involve a sizable proportion of mismatched agents for long durations. For simplicity, suppose that at each step of the dynamics, a blocking pair is chosen uniformly at random.\footnote{\label{footnote_other_probability}Similar results hold for other dynamics. As shown in the Online Appendix, the same insights are obtained for the dynamics where at each step, a randomly-chosen agent---selected uniformly at random among those who have at least one blocking partner---forms a match with her best blocking partner. Analogous results also emerge in various markets with cardinal preferences, where a blocking pair generating a higher total surplus, or a total surplus gain compared to the previous match, is more likely to be formed.}

First, we consider markets with multiple stable matchings. For each market size $n$, we simulate $n \times n$ markets with uniformly random preferences and select 1,000 markets with multiple stable matchings.  In each of these markets, we examine two stable matchings. One is an extremal stable matching, which is often targeted by centralized clearinghouses and is thus of special interest; by symmetry, the analysis is identical for firm- and worker-optimal stable matchings. The other is chosen uniformly at random from the set of all stable matchings, including extremal stable matchings. To provide a lower bound on the fragility of stable matchings, we focus on minimal deviations from a stable matching. Specifically, for each of the two stable matchings, we select one almost stable matching uniformly at random. Next, we calculate various statistics, such as the probability of returning to the minimally perturbed stable matching and the expected time to regain stability, by simulating 300 paths to stability. Then, for each market size and each perturbed stable matching, we compute the expected value of these statistics across all markets in the sample.

\begin{figure}[!ht]
\centering
\begin{subfigure}[t]{.4\textwidth}
\centering
\includegraphics[width=\linewidth]{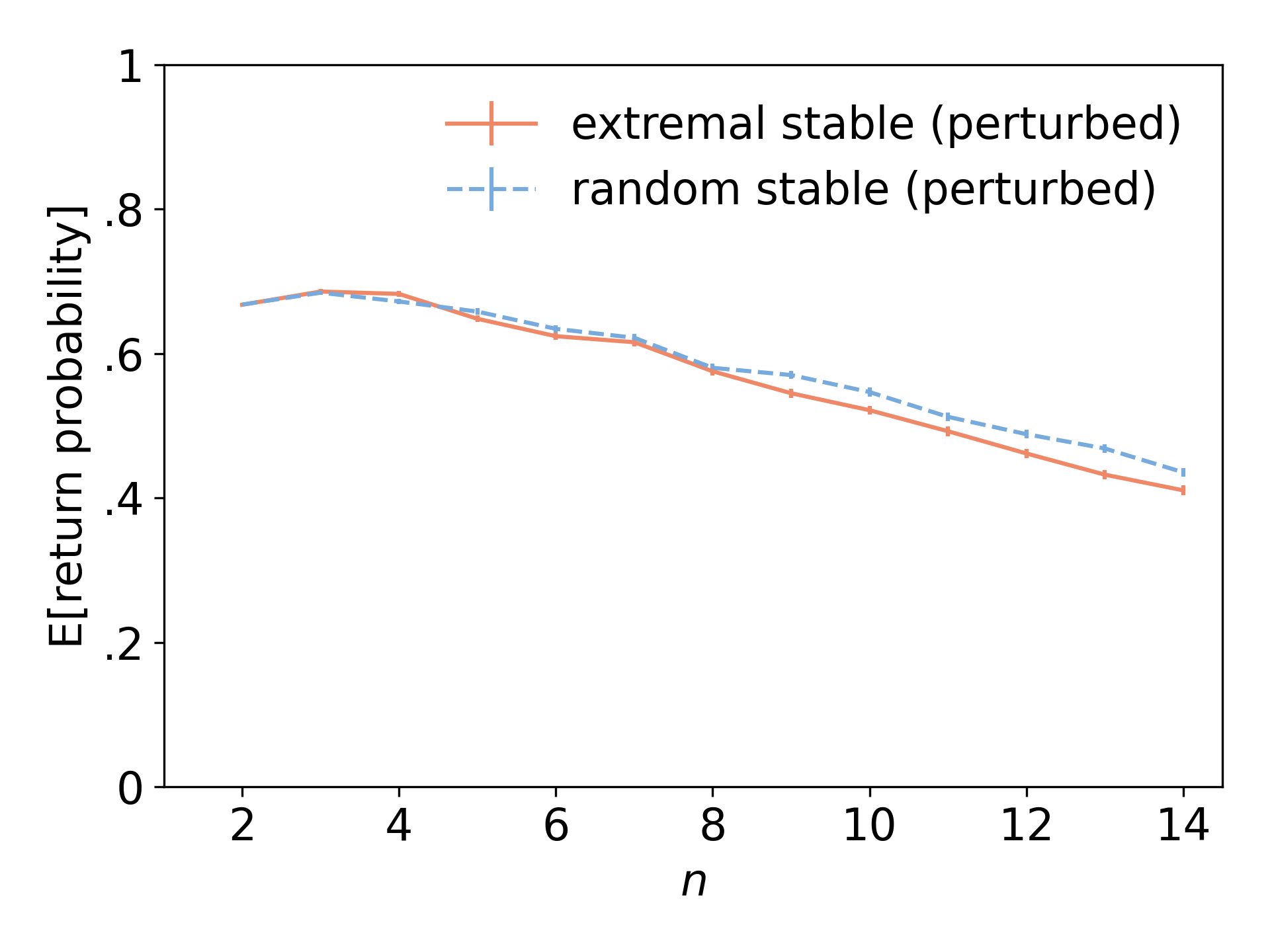}
        \caption{Return probability}
\label{figure_multiple_a}
\end{subfigure}
\begin{subfigure}[t]{.4\textwidth}
\centering
\captionsetup{justification=centering}
\includegraphics[width=\linewidth]{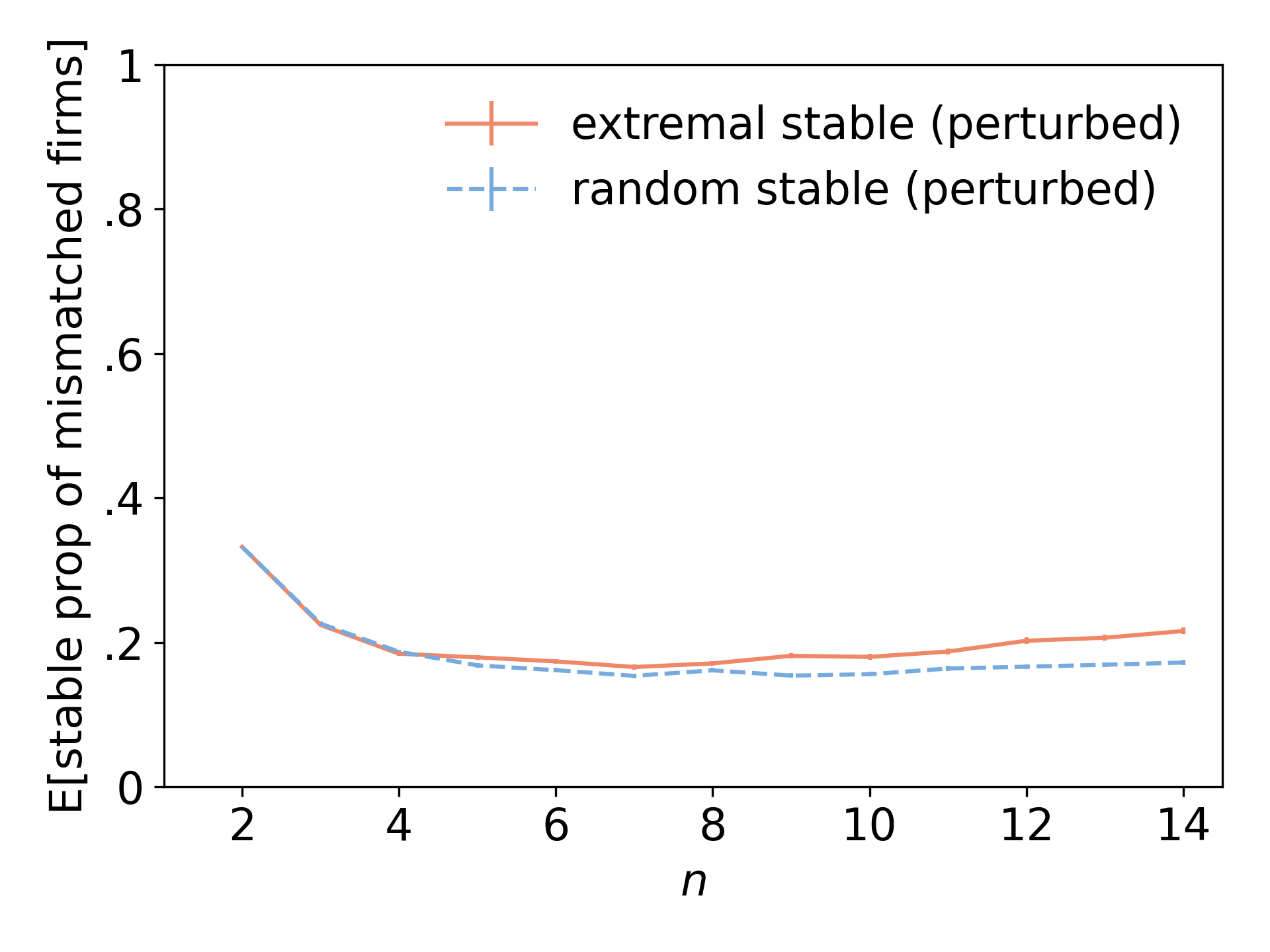}
        \caption{Proportion of mismatched firms/workers in ultimate stable matching}
\label{figure_multiple_b}
\end{subfigure}
\begin{subfigure}[t]{.4\textwidth}
\centering
\includegraphics[width=\linewidth]{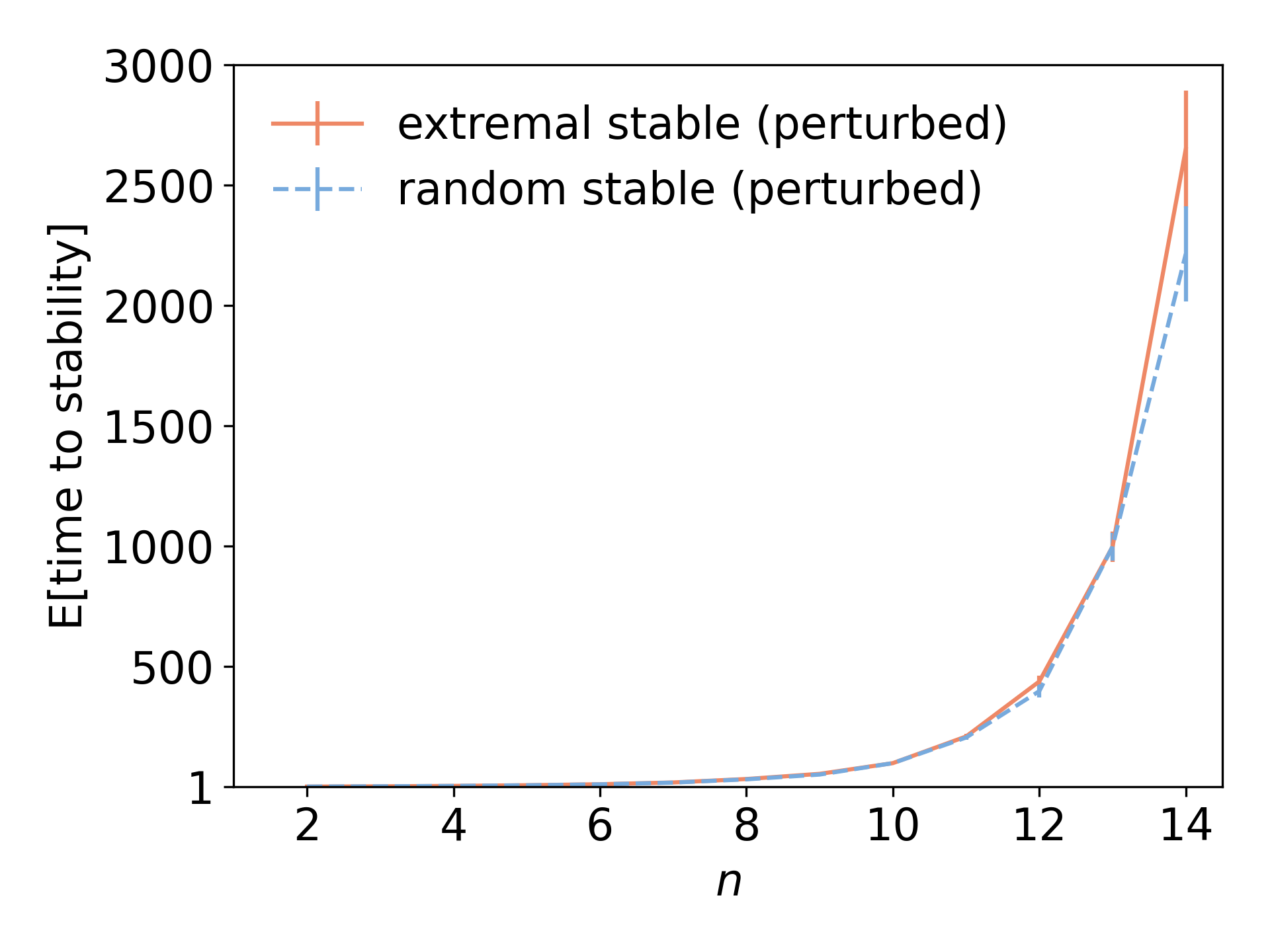}
        \caption{Time to stability}
\label{figure_multiple_c}
\end{subfigure}
\begin{subfigure}[t]{.4\textwidth}
\centering
\includegraphics[width=\linewidth]{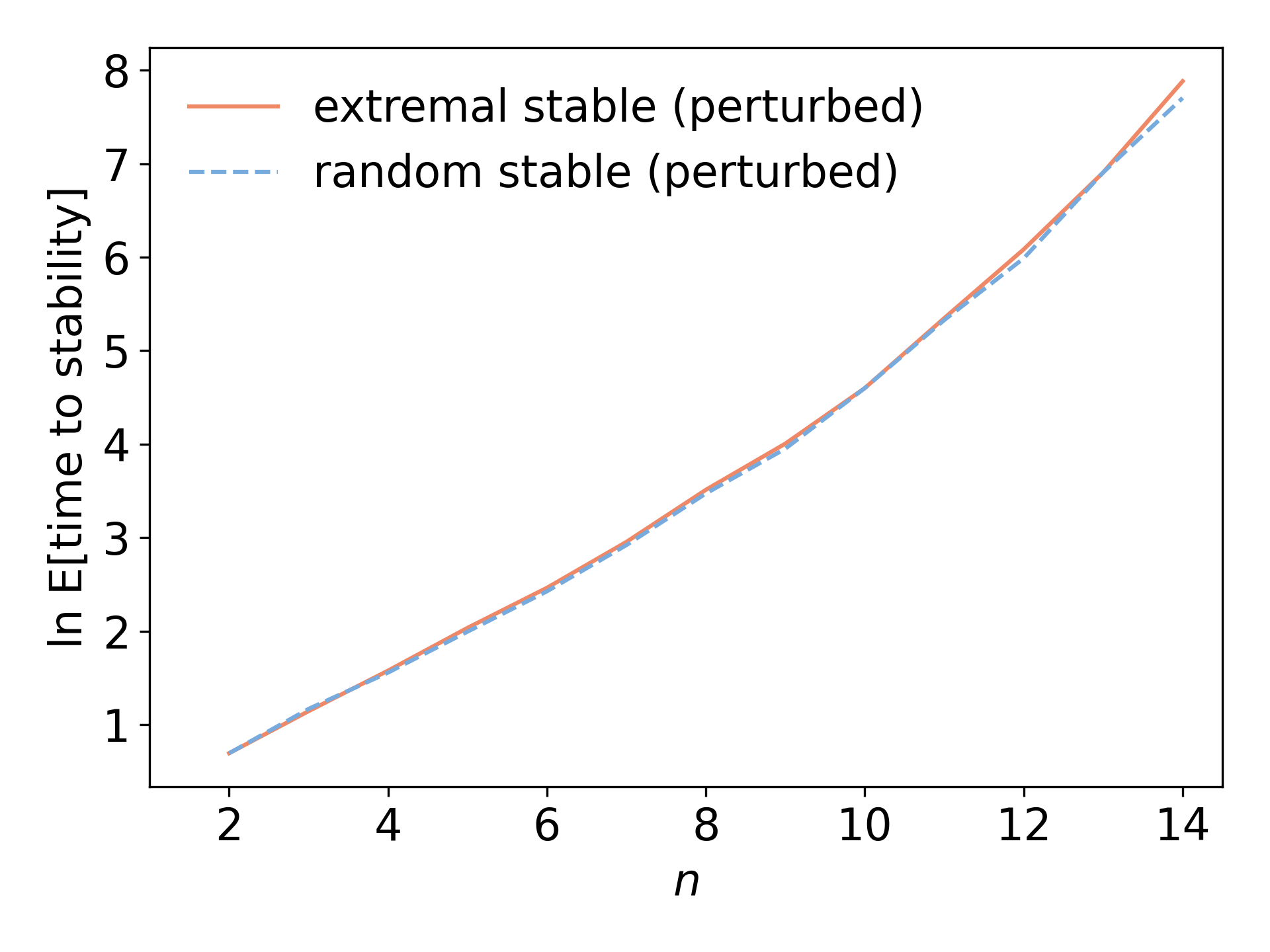}
        \caption{$\ln\left[\text{Time to stability}\right]$}
\label{figure_multiple_d}
\end{subfigure}
\begin{subfigure}[t]{.4\textwidth}
\centering
\captionsetup{justification=centering}
\includegraphics[width=\linewidth]{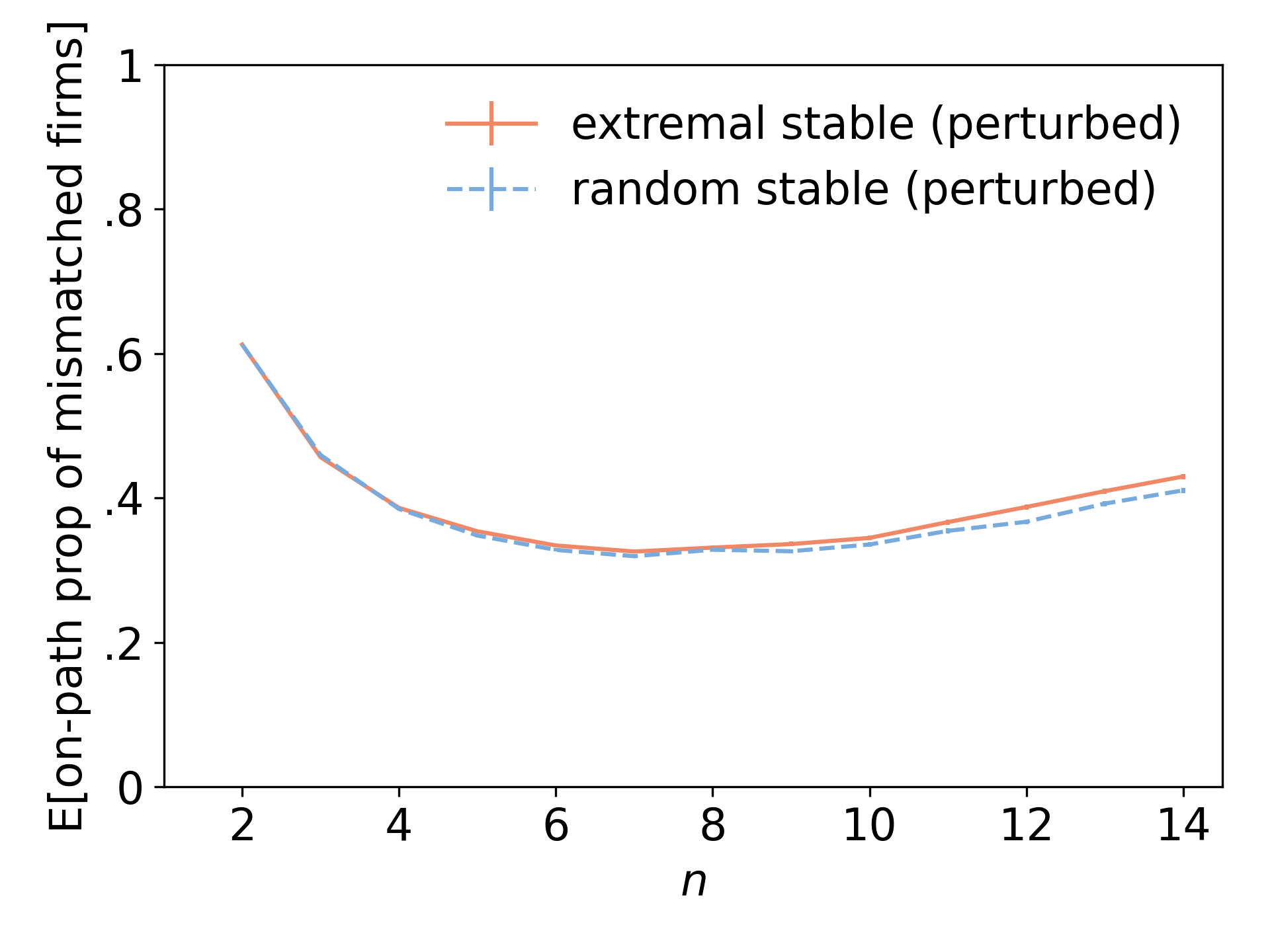}
        \caption{On-path average proportion of mismatched firms/workers}
\label{figure_multiple_e}
\end{subfigure}
\captionsetup{justification=centering}
\caption{Random $n \times n$ markets with multiple stable matchings}
\label{figure_multiple}
\end{figure}

Figure~\ref{figure_multiple} highlights fragility aspects of stable matchings. Even a minimally perturbed stable matching takes an extremely---possibly exponentially---long time to attain a stable matching, not necessarily the perturbed one; see Panels~(\ref{figure_multiple_c}) and~(\ref{figure_multiple_d}) that illustrate the stabilization time and its natural logarithm, respectively. The stabilization process is not only slow but also involves many mismatched agents for long periods, as shown in Panel~(\ref{figure_multiple_e}), which exhibits the on-path average proportion of agents that have different partners compared to the perturbed stable matching. These observations confirm and refine our intuition, discussed in Section~\ref{section_exponential}, regarding the prevalence of ``destabilizing'' rematching opportunities. 

Since decentralized interactions often lead the market away from stability, there is no guarantee whatsoever that the perturbed stable matching returns back or even converges to a nearby, similar stable matching. Panel~(\ref{figure_multiple_a}) shows that indeed the corresponding return probabilities are eventually decreasing, despite a slight non-monotonicity. In markets of size $n=14$, extremal and random stable matchings, when perturbed minimally, return with a probability of only 40-50\%, with the remaining probability converging to other stable matchings. These other stable matchings are considerably different, as demonstrated in Panel~(\ref{figure_multiple_b}), which displays a sizable and non-vanishing proportion of agents whose ultimate stable partners differ from their initial stable partners. Interestingly, extremal stable matchings---a common target of centralized clearinghouses---are more fragile than random stable matchings.\footnote{The Online Appendix presents two examples that analyze the fragility of stable matchings within a given market. Example 5 provides a market with two completely different stable matchings and shows that one of them is fragile with respect to arbitrary perturbations. Example 6 reports a market in which all stable matchings are fragile and extremal stable matchings are most fragile.}\textsuperscript{,}\footnote{\cite{roth1989two} and \cite*{fernandez2022centralized} study outcomes of a centralized clearinghouse in the presence of uncertainty and show that unstable outcomes can emerge in equilibrium.} These findings sharpen the results in Section~\ref{section_characterization}, particularly Corollary~\ref{corollary_characterization}.

Second, we focus on markets with a unique stable matching. We simulate 1,000 such markets for each market size. For each market, we select one almost stable matching uniformly at random and simulate 300 paths to the stable matching. In addition to balanced $n \times n$ markets, we also consider imbalanced $n \times (n+k)$ markets, both with a fixed imbalance $k \in \{1,2,3\}$ and an increasing imbalance $k=n$. This is because for markets with \textit{uniformly random preferences}, imbalanced markets are precisely those markets that have an essentially unique stable matching \citep*{ashlagi2017unbalanced}.

\begin{figure}[H]
\centering
\vspace{-1em}
\begin{subfigure}[t]{.34\textwidth}
\centering
\includegraphics[width=\linewidth]{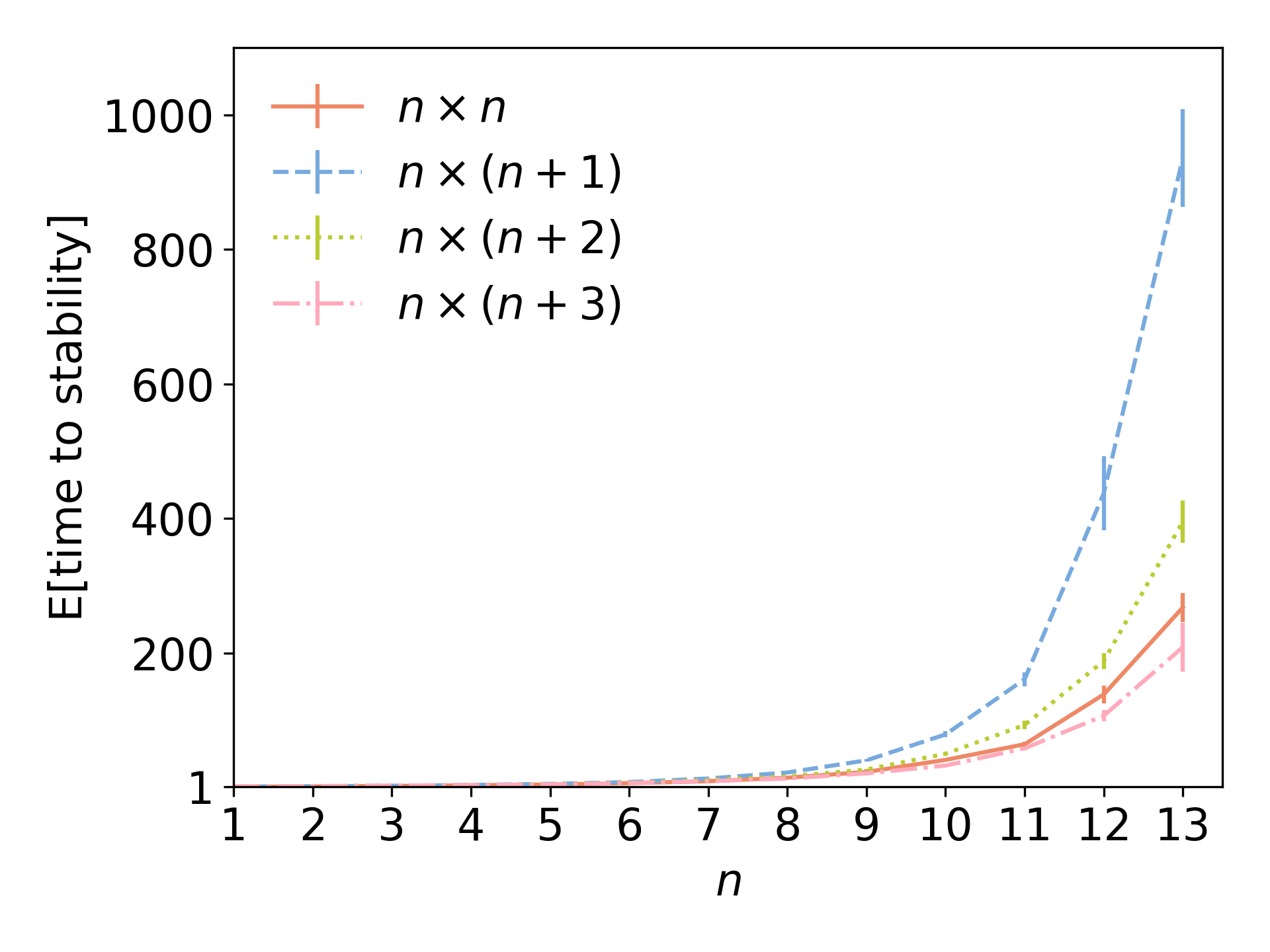}
        \caption{Time to stability}
        \label{figure_unique_fixed_a}
\end{subfigure}
\begin{subfigure}[t]{.34\textwidth}
\centering
\captionsetup{justification=centering}
\includegraphics[width=\linewidth]{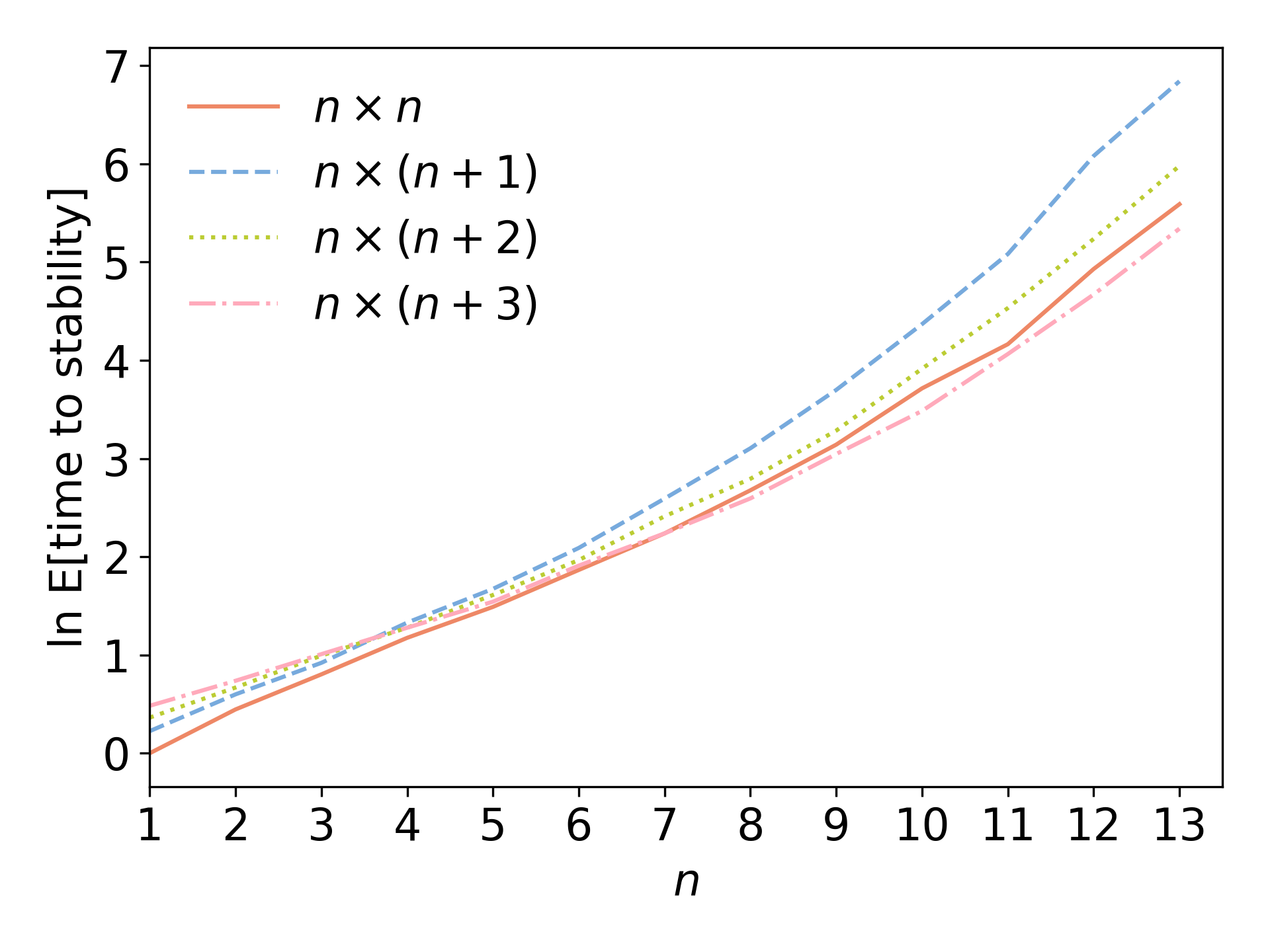}
        \caption{$\ln\left[\text{Time to stability}\right]$}
        \label{figure_unique_fixed_b}
\end{subfigure}
\begin{subfigure}[t]{.34\textwidth}
\centering
\captionsetup{justification=centering}
\includegraphics[width=\linewidth]{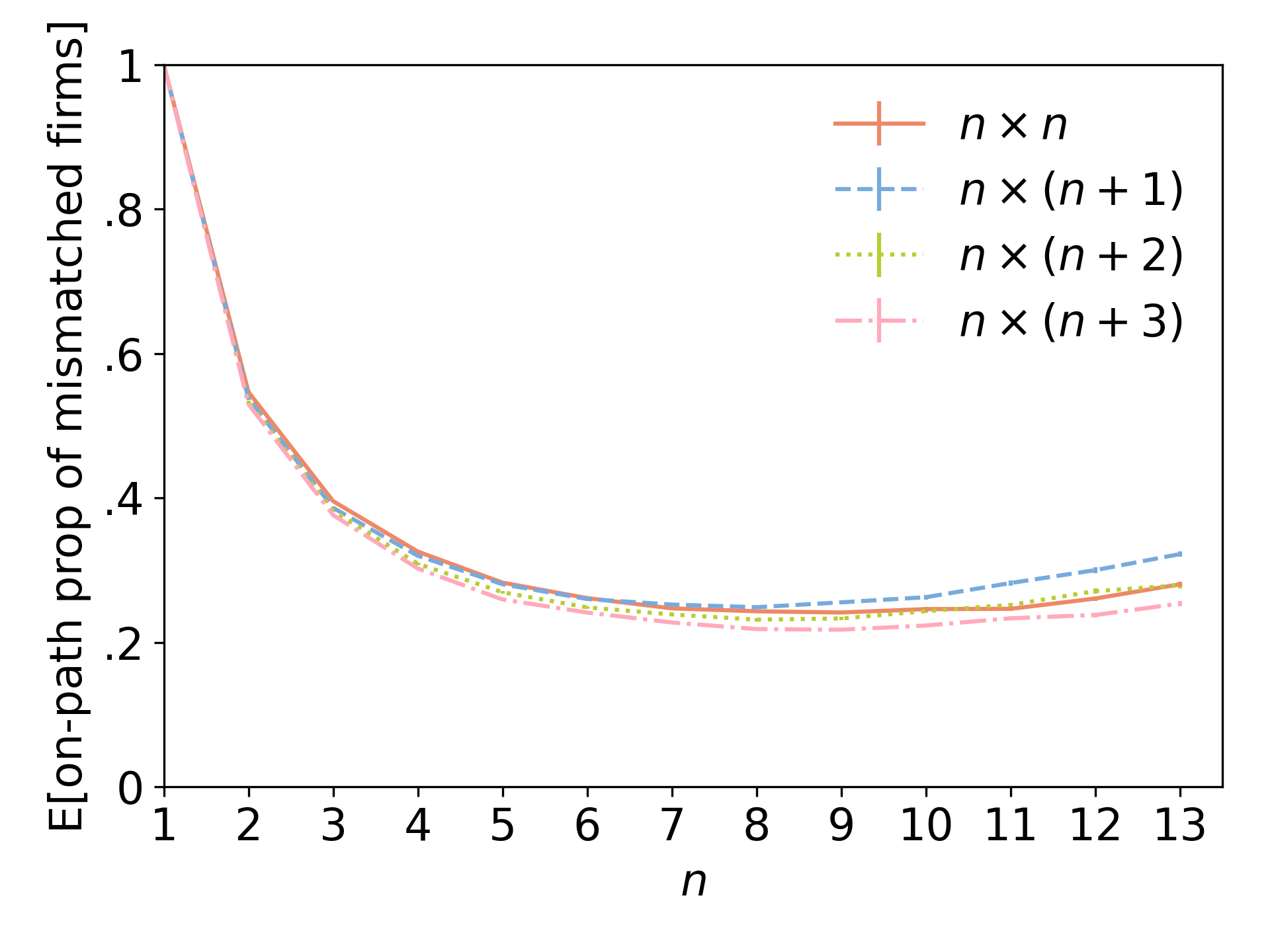}
        \caption{On-path average proportion of mismatched firms}
        \label{figure_unique_fixed_c}
\end{subfigure}
\begin{subfigure}[t]{.34\textwidth}
\centering
\captionsetup{justification=centering}
\includegraphics[width=\linewidth]{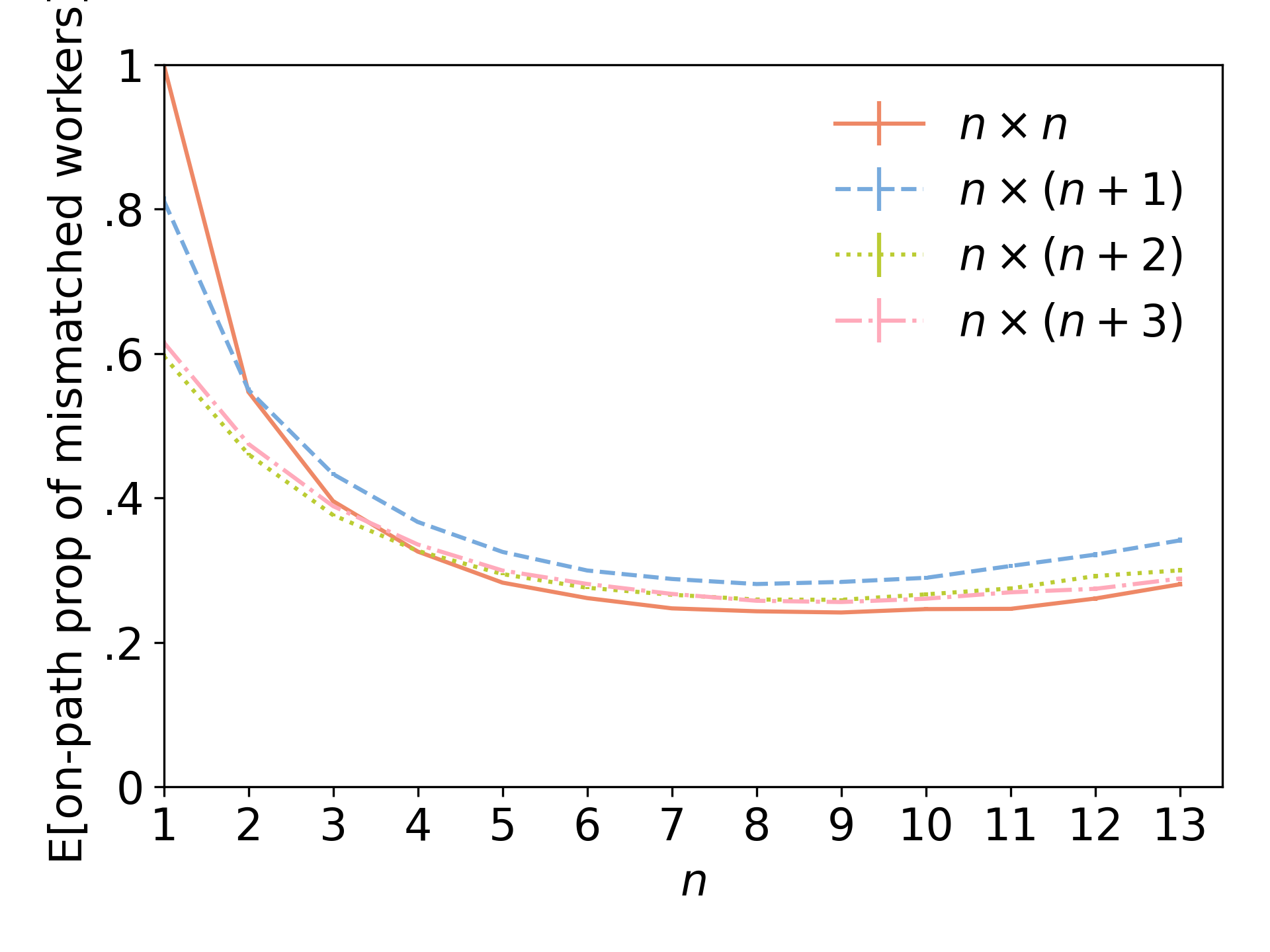}
        \caption{On-path average proportion of mismatched workers}
        \label{figure_unique_fixed_d}
\end{subfigure}
\captionsetup{justification=centering}
\caption{Random $n \times (n+k)$ markets with a unique stable matching, $k \in \{0,1,2,3\}$}
\label{figure_unique_fixed}
%\end{figure}
%\begin{figure}[!ht]
\centering
\begin{subfigure}[t]{.34\textwidth}
\centering
\includegraphics[width=\linewidth]{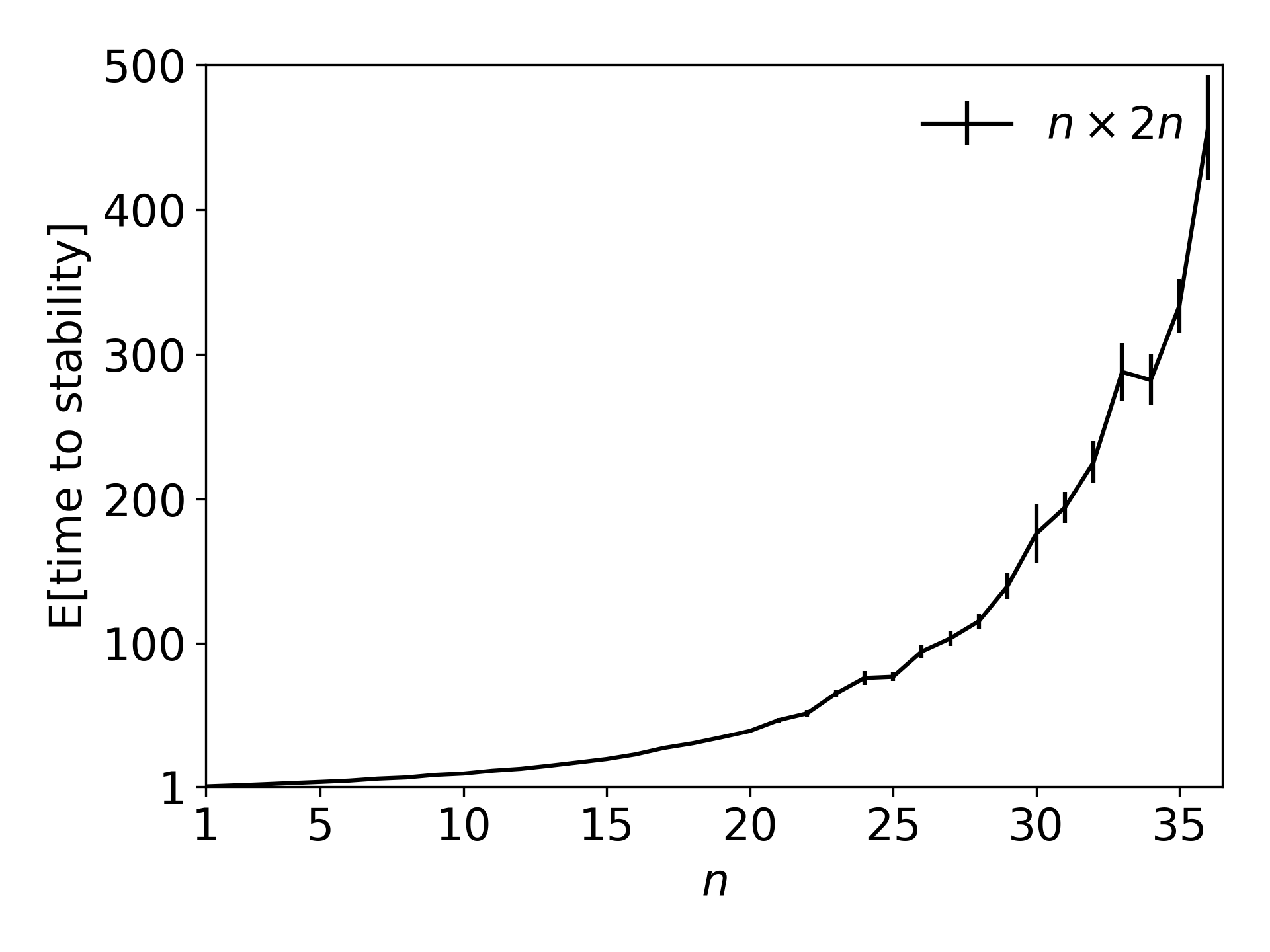}
        \caption{Time to stability}
        \label{figure_unique_increasing_a}
\end{subfigure}
\begin{subfigure}[t]{.34\textwidth}
\centering
\captionsetup{justification=centering}
\includegraphics[width=\linewidth]{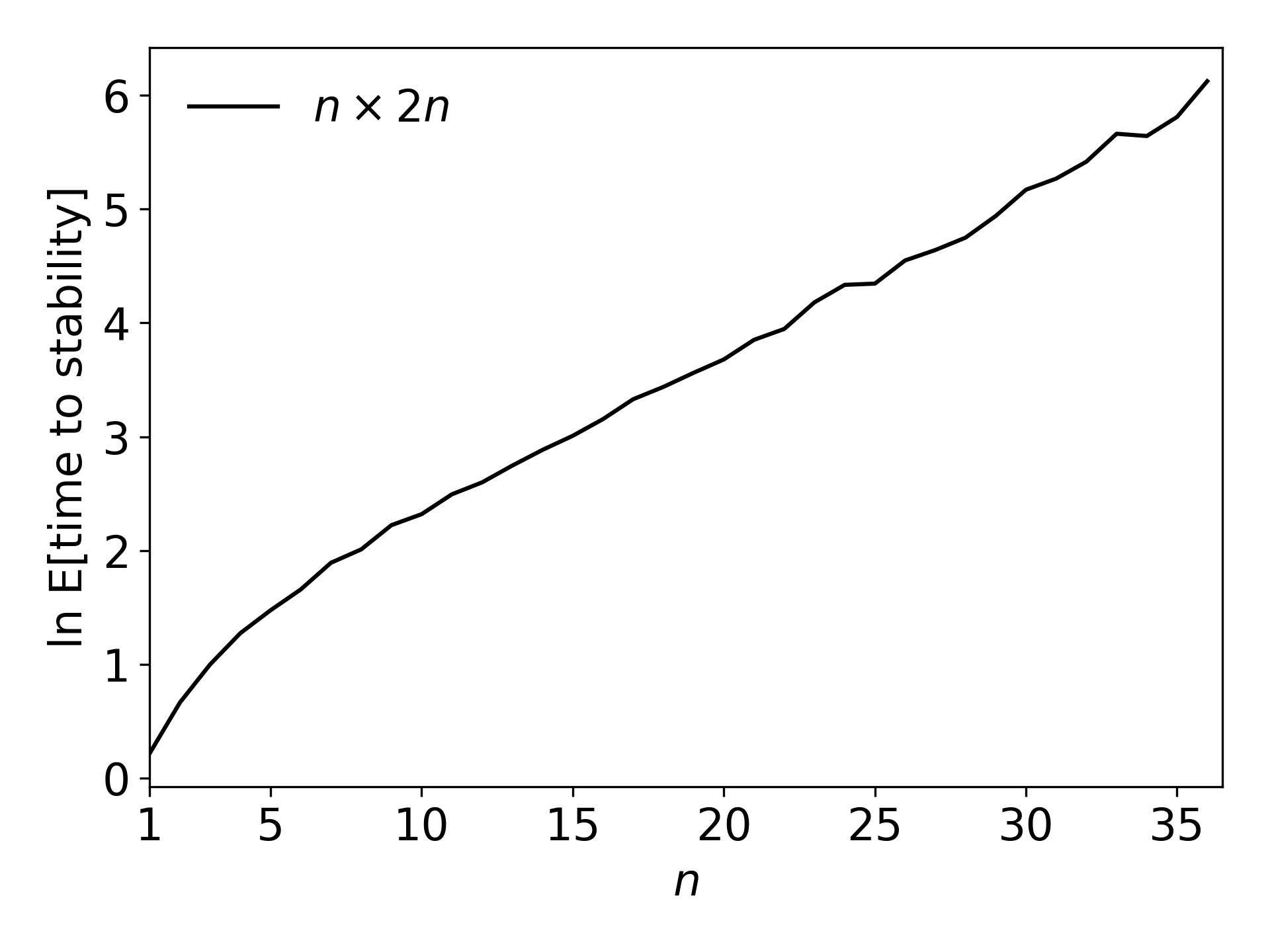}
        \caption{$\ln\left[\text{Time to stability}\right]$}
        \label{figure_unique_increasing_b}
\end{subfigure}
\begin{subfigure}[t]{.34\textwidth}
\centering
\captionsetup{justification=centering}
\includegraphics[width=\linewidth]{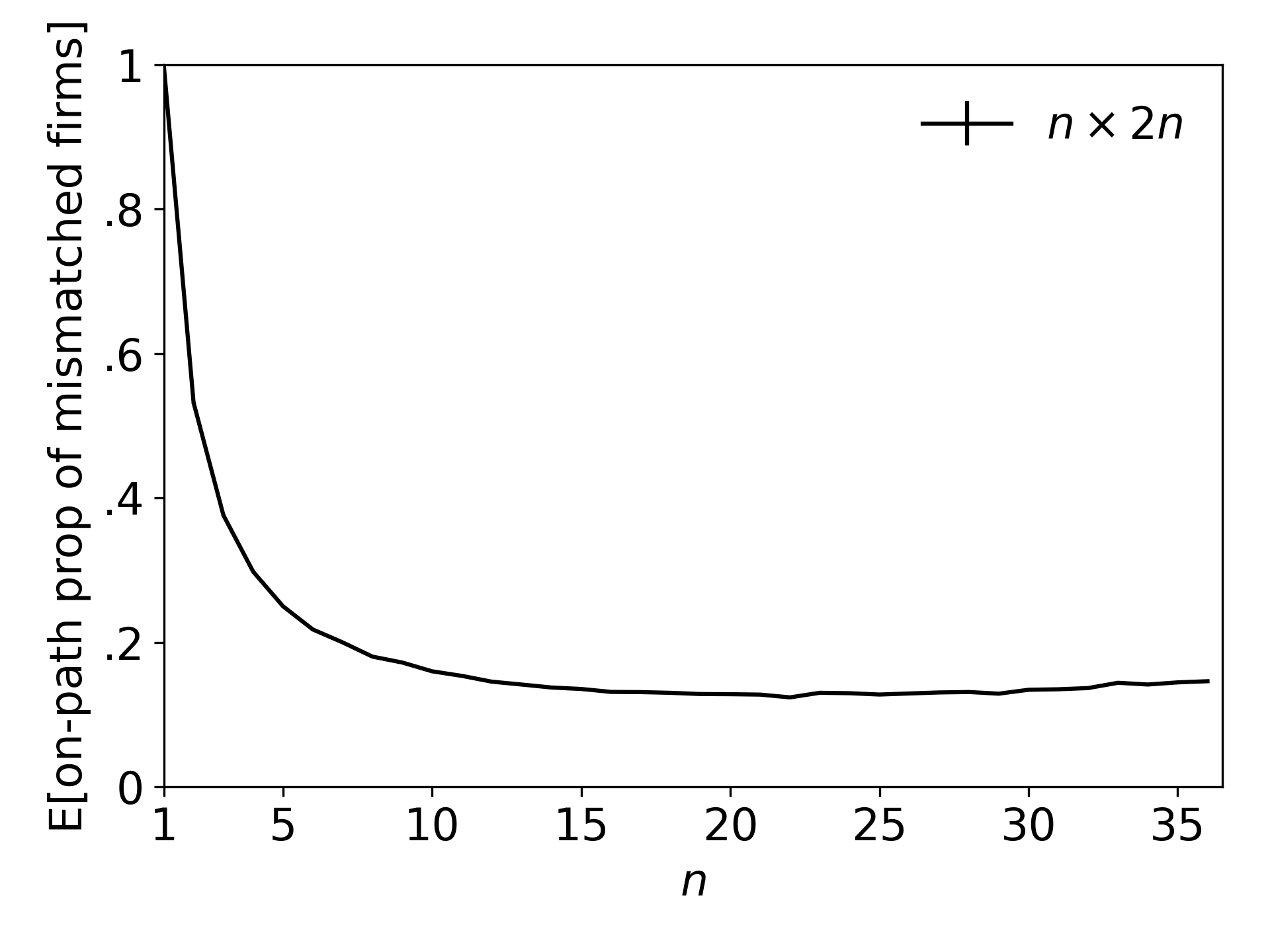}
        \caption{On-path average proportion of mismatched firms}
        \label{figure_unique_increasing_c}
\end{subfigure}
\begin{subfigure}[t]{.34\textwidth}
\centering
\captionsetup{justification=centering}
\includegraphics[width=\linewidth]{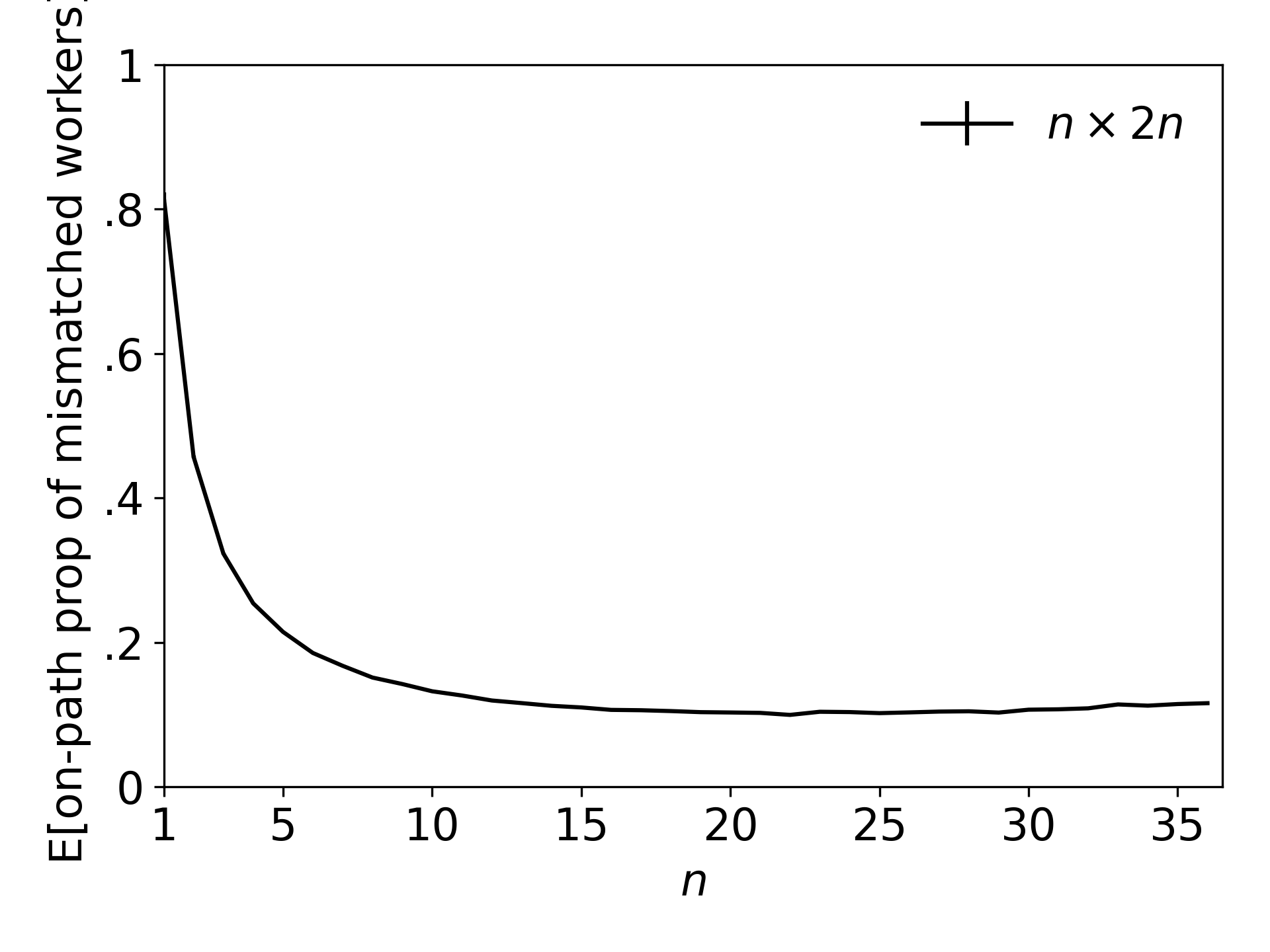}
        \caption{On-path average proportion of mismatched workers}
        \label{figure_unique_increasing_d}
\end{subfigure}
\captionsetup{justification=centering}
\caption{Random $n \times 2n$ markets with a unique stable matching}
\label{figure_unique_increasing}
\end{figure}

Figures~\ref{figure_unique_fixed} and~\ref{figure_unique_increasing} show that unique stable matchings are also fragile, regardless of the size of the imbalance. Although the stabilization process is faster compared to markets with multiple stable matchings, it still lasts for a substantial, and possibly exponentially large, amount of time; see panels (\ref{figure_unique_fixed_a}), (\ref{figure_unique_fixed_b}), (\ref{figure_unique_increasing_a}), and (\ref{figure_unique_increasing_b}). Furthermore, the dynamics typically entail long periods with a considerable proportion of agents being mismatched, as illustrated by panels (\ref{figure_unique_fixed_c}), (\ref{figure_unique_fixed_d}), (\ref{figure_unique_increasing_c}), and (\ref{figure_unique_increasing_d}).

\section{Discussion}

This paper demonstrates the fragility of stable matchings. Even for slight deviations from stability, the stabilization dynamics typically take a very long time, involving many mismatched market participants for extended periods. In markets with multiple stable matchings, stabilization often does not lead back to the slightly perturbed matching or even to a close, similar stable matching. In this sense, stable matchings are not ``locally stable.'' Our results hold for a broad and natural class of dynamics, which provide the decentralized foundation of stable matchings and underlie stability as a solution concept, based on the absence of blocking pairs.\footnote{We show the fragility of stable matchings under favorable conditions for their robustness. We focus on complete-information markets and dynamics that ensure convergence to stability. Furthermore, we suppose that an initial deviation from stability, resulting from some arbitrary shock, is slight; and that there are no additional shocks that could hinder the stabilization process. Despite these favorable conditions, stable matchings are fragile. In settings with incomplete information and changing market conditions, even greater hurdles to stability may emerge.}

Stable outcomes, whenever they exist, may also be fragile in other environments. For several decades, researchers have sought to identify conditions, sometimes perceived as restrictive, that guarantee the existence of stable outcomes in various settings such as many-to-one and many-to-many markets, roommate markets, and more recently, supply chains and trading networks (e.g., \citealp{hatfield2023generalized}), among others; or to find outcomes close to stable when no stable outcome exists.\footnote{Little is still known about the core size in many of these settings.} Significant efforts have also been made to establish the decentralized foundation of stable outcomes in these settings, using \cite{roth1990random}-type dynamics analogous to those in this paper. Our results suggest that even stable outcomes, when they exist, may not be inherently robust.

Our findings highlight the importance of further analysis of the robustness of stable matchings; this area has received much less attention than the robustness analysis of equilibria, other solution concepts, and mechanisms in economics. It is possible that stable matchings, as currently defined, are less fragile under certain dynamics. The Online Appendix considers dynamics that, at each step, match multiple blocking pairs whenever possible, along with several other reasonable dynamics. However, none of these dynamics even guarantee convergence to stability, which is arguably a minimal requirement for suitable dynamics. Exploring alternative dynamics appears to be a valuable direction for future research. For alternative dynamics, not necessarily relying on formation of blocking pairs, it could be more appropriate to define stability, as a solution concept, more generally, as being based on the absence of beneficial rematching opportunities under these specific dynamics.\footnote{In the school choice context, stability corresponds to elimination of justified envy and is also regarded as a normative fairness criterion \citep{kojima2010axioms}.} 

There are also other directions for future research. One such direction involves deeper analysis of fragments, a novel concept whose role in various markets is yet to be fully understood, including problems related to algorithmic, probabilistic, and combinatorial aspects. Another direction is using our framework to define numerical robustness measures for stable matchings, which is possible even in ordinal markets, and to compare the fragility levels of different stable matchings within a given market or across markets. Identifying more robust stable matchings may be a promising area.

\appendix
\section*{Appendix A. McVitie and Wilson (1971)'s Algorithm}
\addcontentsline{toc}{section}{Appendix A. McVitie and Wilson (1971)'s Algorithm}
\label{appendix-a}

\cite{gale1962college} showed that the set of stable matchings is non-empty for any matching market. In doing so, they introduced the deferred acceptance algorithm (DA) described below:
\begin{enumerate}
    \item[] Step 1. Each firm proposes to her favorite (acceptable) worker (if any).  Each worker who receives more than one offer, tentatively holds on to his favorite (acceptable) offer (if any), and rejects all other offers.
    \item[] Step $k$. Each firm who was rejected in step $(k-1)$ makes an offer to her favorite (acceptable) worker who has not rejected her yet (if any). Each worker holds his most preferred acceptable offer to date (if any), and rejects all other offers. 
    \item[] Stop. When no further proposals are made, match each firm to the worker (if any) whose proposal he is holding. 
\end{enumerate}

The description above corresponds to the firm-proposing DA. \cite{gale1962college} proved that its output is the firm-optimal stable matching $\mu_F$. Namely, each firm weakly prefers her assigned partner under this matching to the partner she would get in any other stable matching. The worker-proposing DA defined in a similar way produces the worker-optimal stable matching $\mu_W$. In fact, firms and workers have opposing interests with regard to stable matchings: if firms prefer one stable matching over another, workers hold the opposite opinion \citep{knuth1976marriages}. In particular, the firm-optimal stable matching is also worker-pessimal; similarly, the worker-optimal stable matching is firm-pessimal.\footnote{In fact, the set of stable matchings forms a distributive lattice (proved by Conway, see \citealp{knuth1976marriages}).} Furthermore, the set of unmatched agents is the same across all stable matchings, according to the so-called ``rural hospital theorem'' proved by \cite{mcvitie_stable_1970}.

\cite{mcvitie1971stable} developed an alternative sequential version of the firm-proposing DA. In their version, proposals by firms to workers are made one at a time and in an arbitrary order, but not simultaneously as in \cite{gale1962college}'s algorithm. However, both versions result in the same set of proposals and produce the same matching, the firm-optimal stable matching $\mu_F$.
 
Importantly, \cite{mcvitie1971stable} also adapted the firm-proposing DA, by using an operation called \textit{breakmarriage}, to determine the set of all stable matchings starting from the firm-optimal stable matching. In their paper, they considered a balanced market, i.e., $n=m$, with all agents being acceptable. In what follows, we provide extended results---potentially of independent interest---allowing for a market imbalance and unacceptable agents and connect them to our decentralized dynamics; see Lemma~\ref{lemma_breakmarriage} and discussion at the end of this section.\footnote{These extended results can be easily used to generalize our main findings to arbitrary markets.}

For any stable matching $\mu$ and any paired firm $f$ with $\mu(f)=w \in W$, the operation $breakmarriage(\mu, f)$ is defined as follows.  It first ``breaks'' the $(f,w)$-partnership making firm $f$ free and worker $w$ ``semi-free'' in the sense that he will only accept a proposal from a firm he prefers to $f$ \citep{gusfield1987three}; all other agents are initially matched as in $\mu$. Then, the operation forces firm $f$ to propose to the worker following $w$ in her list, thus initiating a sequence of further proposals, rejections, and (tentative) acceptances according to the restarted firm-proposing DA. As there is at most one free firm that has not yet run out of her acceptable alternatives at any time,\footnote{Indeed, all unmatched firms have already run out of their acceptable options in the initial DA instance.} the corresponding process is deterministic and terminates when 
\begin{inparaenum}[(i)]
\item worker $w$ receives a proposal from a firm he prefers to his current partner $f$,
\item or some previously matched firm runs out of her acceptable alternatives,
\item or some previously unmatched worker receives a proposal from an acceptable firm.
\end{inparaenum}
In the first case, the termination is called \textit{successful} and leads to a new stable matching $\mu' \neq \mu$ (Theorem 3 in \citealp{mcvitie1971stable}, also stated as a lemma below).

\begin{namedthm*}{Lemma {\normalfont \citep{mcvitie1971stable}}}
\label{lemma_mcvitie}
Consider any stable matching $\mu$ and any paired firm $f$ with $\mu(f)=w \in W$. If the operation $breakmarriage(\mu, f)$ terminates successfully, the resulting matching $\mu'$ is stable. 
\end{namedthm*}

\begin{proof}
The resulting matching $\mu'$ is individually rational by definition. Therefore, it is sufficient to prove that there are no firm-worker blocking pairs.

Suppose that some firm $f'$ prefers worker $w'$ to her current match under $\mu'$. Then, $w'$ must have had a proposal from $f'$. Since the termination is successful, any worker ends up with his best match among those consisting of all firms who have previously proposed to that worker and the outside option of being unmatched. In particular, $w'$ prefers his current match to $f'$. Therefore, $(f',w')$ does not block $\mu'$, as desired.
\end{proof}

In the following lemma, we interpret the breakmarriage operation in terms of the decentralized dynamics of our interest. This lemma plays an important role in proving Theorem~\ref{theorem_characterization}.

\begin{lemma}
\label{lemma_breakmarriage}
Under the conditions from the previous lemma, suppose that $breakmarriage(\mu, f)$ terminates successfully and generates stable matching $\mu'$. Then, there exists a finite sequence of best blocking pairs that leads from almost stable matching $\mu_{-f}$ to $\mu'$.
\end{lemma}

\begin{proof}
In the course of the breakmarriage operation, there is exactly one free firm---that has not yet run out of her acceptable alternatives---at any time until the successful termination.

At any point of the process, if free $f'$ eventually gets a tentative match with worker $w'$, they form a blocking pair $(f',w')$ for the corresponding matching at that point.

In fact, $(f',w')$ must be the best blocking pair for $w'$. Indeed, for the sake of contradiction,  suppose that instead $(f'',w')$, with $f'' \neq f'$, is the most preferred blocking pair for $w'$. Then, by the algorithm, $w'$ must have had a proposal from $f''$. This contradicts $w'$ agreeing to tentatively match with less desirable $f'$ at the current point of the process.
\end{proof}

The theorem below is the main result in \cite{mcvitie1971stable}, see Theorem 4 in their paper.

\begin{namedthm*}{Theorem {\normalfont \citep{mcvitie1971stable}}}
\label{theorem_mcvitie}
Any stable matching $\mu \neq \mu_F$ can be obtained from the firm-optimal stable $\mu_F$ by successive applications of the breakmarriage operation.
\end{namedthm*}

\begin{proof}
Consider any stable matching $\mu \neq \mu_F$. Then, there is a firm $f$ whose partner under $\mu$ is not the same as under $\mu_F$. In fact, by the rural hospital theorem (requiring the set of unmatched agents to be the same across all stable matchings), both partners $\mu(f)$ and $\mu_F(f)$ are workers. Apply $breakmarriage(\mu_F, f)$.

We will show that no previously matched firm $f' \in \mu(W)$ can propose to a worse alternative than $\mu(f')$, and thus the process will terminate successfully.\footnote{First, the condition (ii) of ``unsuccessful'' termination would trivially fail. Second, if (iii) some previously unmatched worker $w''$ received a proposal from an acceptable firm $f''$, $f'' \in \mu(W)$ would prefer $w''$ to $\mu(f'')$, and thus stable $\mu$ would be blocked by $(f'', w'')$, leading to a contradiction.} Suppose by contradiction that $f' \in \mu(W)$ is the first firm who is rejected by her partner under $\mu$, i.e., worker $\mu(f')$. Then, worker $\mu(f')$ has already received a proposal from a preferred firm $f''$. We must have $f'' \in \mu(W)$ as well since all firms outside $\mu(W)$ have already run out of their acceptable options in the initial DA instance. By our assumption, $f''$ has not yet proposed to worker $\mu(f'')$; otherwise, $\mu(f'')$ rejected $f''$ before $\mu(f')$ rejected $f'$. Therefore, $f''$ strictly prefers $\mu(f')$ to $\mu(f'')$. This implies that stable $\mu$ is blocked by $(f'', \mu(f'))$, leading to a contradiction.

Therefore, $breakmarriage(\mu_F, f)$ outputs a stable matching in which no previously matched firm has a worse partner than under $\mu$. If the resulting stable matching is not $\mu$, we can iteratively apply the breakmarriage operation until we reach $\mu$ after a finite number of applications. Indeed, the breakmarriage operation always makes all firms weakly worse off with some being strictly worse off.
\end{proof}

Notably, using the same arguments as in the proof, we can obtain the worker-optimal stable $\mu_W$ by successive applications of breakmarriage operations starting from any stable matching, not only from the firm-optimal stable $\mu_F$. Furthemore, by symmetric arguments as in this section, from the worker-optimal stable $\mu_W$, we can get any stable matching $\mu \neq \mu_W$, now by switching to the worker-proposing DA and using symmetrically defined breakmarriage operations with the roles of firms and workers reversed. 

In essence, we can extend the above theorem well beyond initial $\mu_F$: by using appropriate versions of DA and breakmarriage operations, any stable matching $\mu'$ can be obtained from any stable matching $\mu \neq \mu'$. We use this observation, together with our interpretation of breakmarriage operations in Lemma~\ref{lemma_breakmarriage}, to prove Theorem~\ref{theorem_characterization}, one of our main results. 

\section*{Appendix B. Proof of Lemma~\ref{lemma_definition} and Theorem~\ref{theorem_characterization}}
\addcontentsline{toc}{section}{Appendix B. Proof of Lemma~\ref{lemma_definition} and Theorem~\ref{theorem_characterization}}
\label{appendix-b}

In this section, we prove Lemma~\ref{lemma_definition} and Theorem~\ref{theorem_characterization} stated in Section~\ref{section_characterization}.

\begin{delayedproof}{lemma_definition}
Up to relabeling, suppose that fragment $(\bar{F}, \bar{W})$ of size $k<n$ is induced by matching $\bar{\mu}$ such that $\bar{\mu}(f_i)=w_i$ for all $i \in [k]$.

Assume by contradiction that for some stable matching $\mu$ in the original market, we have $\mu(\bar{F}) \neq \bar{W}$. Then, there exists $\bar{w} \in \bar{W}$ with $\mu(\bar{w}) \notin \bar{F}$. Up to relabeling, $\bar{w}=w_1$.\bigskip

\noindent \textit{Step 1.} By definition of $\bar{\mu}$,
\begin{equation*}
\bar{\mu}(w_1)=f_1 \succ_{w_1} \mu(w_1) \notin \bar{F}.
\end{equation*}

\noindent By stability of $\mu$,
\begin{equation*}
{\mu}(f_1) \succ_{f_1} w_1 = \bar{\mu}(f_1).
\end{equation*}

\noindent By definition of $\bar{\mu}$, ${\mu}(f_1) \in \bar{W}$. Also, ${\mu}(f_1) \neq w_1$. Then, up to relabeling, we can assume ${\mu}(f_1)=w_2$. Therefore,
\begin{equation*}
{\mu}(f_1) = w_2 \succ_{f_1} w_1 = \bar{\mu}(f_1).
\end{equation*}

\noindent \textit{Step 2.} By stability of $\bar{\mu}$,
\begin{equation*}
\bar{\mu}(w_2)=f_2 \succ_{w_2} f_1 = \mu(w_2).
\end{equation*}

\noindent By stability of $\mu$,
\begin{equation*}
{\mu}(f_2) \succ_{f_2}  w_2 = \bar{\mu}(f_2).
\end{equation*}

\noindent By definition of $\bar{\mu}$, ${\mu}(f_2) \in \bar{W}$. In addition, ${\mu}(f_2) \neq w_1, w_2$.\footnote{Indeed, since $\mu(f_1)=w_2$, we must have ${\mu}(f_2) \neq w_2$.} As before, up to relabeling, assume ${\mu}(f_2)=w_3$.\bigskip

We can proceed recursively, setting $\mu(f_i)=w_{i+1}$ for $i < k$, until we reach the final step.\bigskip

\noindent \textit{Step k.} By stability of $\bar{\mu}$,
\begin{equation*}
\bar{\mu}(w_k)=f_k \succ_{w_k} f_{k-1} = \mu(w_k).
\end{equation*}

\noindent By stability of $\mu$,
\begin{equation*}
{\mu}(f_k) \succ_{f_k}  w_k = \bar{\mu}(f_k).
\end{equation*}

\noindent By definition of $\bar{\mu}$, ${\mu}(f_k) \in \bar{W}$. However, by our previous steps, ${\mu}(f_k) \neq w_1, w_2, \ldots, w_k$, i.e., ${\mu}(f_k) \notin \bar{W}$. This delivers a contradiction, and the result follows.
\end{delayedproof}

The following observation, stated as Lemma~\ref{lemma_submarket}, is used in our proof of Theorem~\ref{theorem_characterization} below.

\begin{lemma}
\label{lemma_submarket}
Suppose that there are no non-trivial fragments in the original market. If there is a trivial fragment, then after removing agents forming this fragment from the original market, the remaining market has no non-trivial fragments as well.
\end{lemma}
\begin{proof}
Suppose, towards a contradiction, that there is a non-trivial fragment in the remaining market. It is straightforward to verify that when merged with the excluded trivial fragment, it extends to a non-trivial fragment in the original market.
\end{proof}

\begin{delayedproof}{theorem_characterization}
It suffices to prove $(ii) \Rightarrow (i) \Rightarrow (iii) \Rightarrow (ii)$.\bigskip

\noindent \textbf{(ii) $\bm \Rightarrow$ (i):} Consider any unstable matching $\lambda$. Then, by $(ii)$, it is sufficient to show that $\lambda$ can attain an almost stable matching. To prove this, we use the recent stabilization result due to \cite{ackermann2011uncoordinated} from the computer science literature, which strengthens the classical result of \cite{roth1990random}: for any unstable matching, there exists a finite sequence of best blocking pairs that leads to a stable matching. 

By this stabilization result, matching $\lambda$ can reach some stable matching $\mu$. Then, $\lambda$ can attain an almost stable matching that is one blocking pair away from $\mu$. Indeed, we can terminate the path to $\mu$ right before the last step at which the last blocking pair is satisfied.\bigskip

\noindent \textbf{(i) $\bm\Rightarrow$ (iii):} Suppose, for contradiction, that there is a non-trivial fragment, induced by matching $\bar{\mu}$. Focus on any unstable matching $\lambda$ that agrees with $\bar{\mu}$ over the fragment. Then, any stable matching that $\lambda$ can attain must also agree with $\bar{\mu}$ over the fragment.

Since the fragment is non-trivial, there are some stable matchings that disagree with $\bar{\mu}$ over the fragment; and thus they cannot be reached from $\lambda$. This yields a contradiction.\bigskip

\noindent \textbf{(iii) $\bm \Rightarrow$ (ii):} We prove this implication by induction on the size $n=|F|=|W|$ of markets with no non-trivial fragments. For $n=1$, the result holds immediately. Assume that it holds for all such markets of size at most $(n-1)$, where $n \geq 2$. 

Consider any market of size $n$, again with no non-trivial fragments. The remaining proof relies on the following two observations.\bigskip

\noindent \textit{Observation 1.} Consider an arbitrary almost stable matching $\lambda$ which corresponds to some stable matching $\mu$. We show that $\lambda$ can reach either any stable matching---as desired---or any other almost stable matching for stable matching $\mu$. A proof of this crucial observation is given at the end of this section.

\bigskip

\noindent \textit{Observation 2.} Lemma~\ref{lemma_breakmarriage} in Appendix~\hyperref[appendix-a]{A} establishes the connection between the decentralized dynamics of our interest and breakmarriage operations in \cite{mcvitie1971stable}'s algorithm used to find the set of all stable matchings; for the algorithm's description and more details on the connection, see Appendix~\hyperref[appendix-a]{A}. Roughly speaking, any breakmarriage operation, applied to a stable matching, first ``divorces'' a stable pair of agents and then ``restarts'' some version of DA. Lemma~\ref{lemma_breakmarriage} shows that for all relevant breakmarriage operations, the restarted DA algorithm following breakups can be emulated by our decentralized dynamics.\bigskip

By combining these two observations, we obtain the desired result. Indeed, consider an arbitrary almost stable matching $\lambda$ that corresponds to some stable matching $\mu$. Observation 1 allows to attain any almost stable matching---or, equivalently, to ``divorce'' any stable pair---for $\mu$; assuming we have not shown that $\lambda$ can reach any stable matching. This, together with Observation 2, implies that the dynamics can replicate any relevant breakmarriage operation for $\mu$ and attain a new stable matching $\mu'$. Similar to the proof of  $(ii) \Rightarrow (i)$, since $\lambda$ can attain stable $\mu'$, one of its almost stable matchings $\lambda'$ can also be attained. 

Next, we apply the same reasoning for almost stable $\lambda'$ corresponding to stable $\mu'$. By Observation 1, any almost stable matching for $\mu'$ can be attained; assuming we have not demonstrated that $\lambda'$---and, consequently, $\lambda$---can reach any stable matching. As before, in conjunction with Observation 2, this allows to attain a new stable matching $\mu''$. By proceeding recursively and taking into account our discussion in the last two  paragraphs of Appendix~\hyperref[appendix-a]{A}, it follows that the decentralized dynamics can emulate all relevant breakmarriage operations in \cite{mcvitie_stable_1970}'s algorithm. Specifically, $\lambda$ can attain any stable matching, as desired. Below, we complete the proof by establishing Observation 1.

\bigskip

\noindent \textit{Proof of Observation 1.} 
Consider an arbitrary almost stable matching $\lambda$. Up to relabeling, $\lambda =  \mu_{-f_n}$ for stable matching $\mu$, where $\mu(f_i)=w_i$ for any $i \in [n]$. We want to show that $\mu_{-f_n}$ can attain either any stable matching or any other almost stable matching $\mu_{-f_i}$, $i<n$.\bigskip

\noindent \textit{Step 1.} Consider $\bar{F} \equiv \{f_i\}_{i \leq n-1}$ and $\bar{W} \equiv \{w_i\}_{i \leq n-1}$. Suppose first that  $\bar{F}$ and $\bar{W}$ form a fragment. This fragment must be trivial due to the absence of non-trivial fragments. From the definition of a trivial fragment, it follows that stable matching $\mu$ is the unique stable matching. Then, $\mu_{-f_n}$ can easily attain unique stable $\mu$---and therefore any stable matching, as desired---say, by the stabilization result as in the proof of  $(ii) \Rightarrow (i)$; or, more directly, by allowing $f_n$ to choose her best blocking pair $(f_n, w_n)$.

Henceforth, suppose that $(\bar{F}, \bar{W})$ is not a fragment. By definition, there is an ``inside'' agent $a \in \bar{F} \cup \bar{W}$ that prefers some ``outside'' agent $a' \notin  \bar{F} \cup \bar{W}$ to her stable partner $\mu(a)$.
\begin{enumerate}
    \item If $a \in \bar{F}$, then, up to relabeling, $a=f_{n-1}$. Also, $a'=w_{n}$ is the only ``outside'' worker. Thus, $w_n \succ_{f_{n-1}} \mu(f_{n-1})=w_{n-1}$. By stability of $\mu$, we have  $f_n = \mu(w_n) \succ_{w_n} f_{n-1}$.
    
    Note that $\mu_{-f_n}$ can attain the new almost stable matching $\mu_{-f_{n-1}}$ by sequentially satisfying two blocking pairs, $(f_{n-1}, w_n)$ and $(f_n, w_n)$. This path corresponds to agents $f_{n-1}$ and $w_n$ that successively choose their best blocking pairs, first $f_{n-1}$ and then $w_n$. The corresponding blocking pairs are indeed the best by stability of $\mu$.
    
    \item Otherwise, if $a \in \bar{W}$, then, up to relabeling, $a=w_{n-1}$. Now, $a'=f_{n}$ is the only ``outside'' firm. By symmetric arguments, $\mu_{-f_n}$ can attain $\mu_{-f_{n-1}}$ through the path on which first $w_{n-1}$ and then $f_n$ select their best blocking pairs.
\end{enumerate}

\noindent In either case, $\mu_{-f_{n}}$ can attain the new almost stable matching $\mu_{-f_{n-1}}$.\bigskip

\noindent \textit{Step 2.} Redefine subsets $\bar{F} \equiv \{f_i\}_{i \leq n-2}$ and $\bar{W} \equiv \{w_i\}_{i \leq n-2}$. Suppose first that $\bar{F}$ and $\bar{W}$ form a fragment. Then, it must be trivial. By  the definition of a trivial fragment, agents within the fragment must be matched in accordance with $\mu$ for any stable matching in the original market. By Lemma~\ref{lemma_submarket}, the remaining market without those agents has no non-trivial fragments. When ``projected'' to the
remaining market, $\mu_{-f_n}$ remains unstable and, by the induction hypothesis, can attain any stable matching in this smaller market. Since all stable matchings in the original market agree on the removed trivial fragment and continue to be stable when ``projected'' to the remaining market, it implies that matching
$\mu_{-f_n}$
can attain any stable matching in the original market. Therefore, the desired result holds.

Henceforth, suppose that $(\bar{F}, \bar{W})$ is not a fragment. By definition, there is an ``inside'' agent $a \in \bar{F} \cup \bar{W}$ that prefers some ``outside'' agent $a' \notin  \bar{F} \cup \bar{W}$ to her stable partner $\mu(a)$.

\begin{enumerate}
    \item If $a \in \bar{F}$, then, up to relabeling, $a=f_{n-2}$. Since ``outside'' $a' \in \{w_{n-1}, w_{n}\}$, we have either $w_n \succ_{f_{n-2}} \mu(f_{n-2})=w_{n-2}$ or $w_{n-1} \succ_{f_{n-2}} \mu(f_{n-2})=w_{n-2}$.
    
    In the former case, as in \textit{Step 1}, $\mu_{-f_n}$ can attain the new almost stable matching $\mu_{-f_{n-2}}$ through the path on which first $f_{n-2}$ and then $w_n$ pick their best blocking pairs.  
    
    In the latter case, similar to \textit{Step 1}, $\mu_{-f_{n-1}}$ can reach  $\mu_{-f_{n-2}}$ through the path on which first $f_{n-2}$ and then $w_{n-1}$ select their best blocking pairs.  
    
    \item If $a \in \bar{W}$, then, up to relabeling, $a=w_{n-2}$. Now, since ``outside'' $a' \in \{f_{n-1}, f_{n}\}$, we have either $f_n \succ_{w_{n-2}} \mu(w_{n-2})=f_{n-2}$ or $f_{n-1} \succ_{w_{n-2}} \mu(w_{n-2})=f_{n-2}$. 
    
    By symmetric arguments, the new almost stable matching $\mu_{-f_{n-2}}$ can be attained from $\mu_{-f_n}$ either directly or indirectly through $\mu_{-f_{n-1}}$.
\end{enumerate}
\noindent Consequently, in either case, $\mu_{-f_{n}}$ can attain new $\mu_{-f_{n-2}}$, in addition to $\mu_{-f_{n-1}}$.\bigskip

We then proceed recursively until we reach the final step, assuming that we have not established the existence of paths from $\mu_{-f_{n}}$ to all stable matchings in the previous steps.\bigskip

\noindent \textit{Step $(n-1)$.} Redefine subsets $\bar{F} \equiv \{f_1\}$ and $\bar{W} \equiv \{w_1\}$. Suppose first that $\bar{F}$ and $\bar{W}$ form a fragment. Then, it must be trivial. Therefore, $\mu_{-f_{n}}$ can attain any stable matching, as desired, by Lemma~\ref{lemma_submarket} and our induction hypothesis, similar to \textit{Steps} $2,3, \ldots, n-2$.

Henceforth, suppose that $(\bar{F},\bar{W})$ is not a fragment. Then, as before, $\mu_{-f_n}$ can reach the new almost stable matching $\mu_{-f_1}$ either directly or indirectly through some---already shown to be attainable in the previous steps---matchings $\mu_{-f_{n-1}}$, $\mu_{-f_{n-2}}$,\ldots, $\mu_{-f_{2}}$.\bigskip

\noindent To conclude, the above steps guarantee that $\mu_{-f_n}$ can attain either any stable matching or any other almost stable matching $\mu_{-f_i}$, $i<n$. This proves Observation 1, and the desired implication $(ii) \Rightarrow (i)$ follows from our previous arguments.
\end{delayedproof}

\section*{Appendix C. Proof of Theorem~\ref{theorem_exponential}}
\addcontentsline{toc}{section}{Appendix C. Proof of Theorem~\ref{theorem_exponential}}
\label{appendix-c}

In this section, we prove Theorem~\ref{theorem_exponential}, stated in Section~\ref{section_exponential}, in two steps. First, we establish a result---analogous to the theorem---for a class of \textit{non-augmented markets} with a unique stable matching; the identified class is non-empty, as shown by Lemma 6 in the Online Appendix. Second, we employ markets from this class to augment arbitrary markets and prove the desired theorem.

\begin{proposition}
\label{proposition_class}
Focus on $\kappa$-random (best) dynamics. Fix any $\eta \in (0,1)$, $\epsilon \in (0,1]$, and $\kappa \geq 1$. Consider any sequence of markets of size $n \in \mathbb{N}$ with a unique stable matching such that for any market in the sequence,
\begin{inparaenum}[(i)]
\item \label{prop_c_1} any worker except one is preferred by at least an $\eta$-proportion of firms to their stable partners;
\item \label{prop_c_2} and similarly, any firm except one is preferred by at least an $\eta$-proportion of workers to their stable partners.
\end{inparaenum}
Then, any corresponding sequence of $\epsilon$-unstable matchings with probability $1-2^{-\Omega(n)}$ needs $2^{\Omega(n)}$ steps to regain stability.
\end{proposition}

\begin{proof} \textit{Note:} The proof below works irrespective of whether we examine ``standard'' blocking pairs or best blocking pairs. Thereafter, we omit ``best'' if there is no confusion.

Consider the market of size $n$ in the given sequence and denote its unique stable matching by $\mu$. Suppose that firm $\bar{f}$ and worker $\bar{w}$ are the corresponding ``exclusion'' agents, see $(i)$ and $(ii)$ in the premise of the proposition. For any matching $\lambda$ in this market, let 
{\setlength\abovedisplayskip{5pt}
\setlength\belowdisplayskip{5pt}
\begin{equation*}
    \mathcal{S}(\lambda)= |i \in [n] \text{\; such that \;} \lambda(f_i)=\mu(f_i)|,
\end{equation*}}denote the number of firms (or, equivalently, workers) matched with their stable partners. 

Take an arbitrary starting $\epsilon$-unstable matching. To restore stability, at some point of the decentralized process, the market must be arbitrarily close to the stable matching, $\mathcal{S}(\mu)=n$. In fact, the dynamics need to enter the region $\mathcal{S}(\lambda) \geq (1-\zeta)n$ for some $\zeta \leq \epsilon$, and never leave it until convergence; this must hold for arbitrarily small $\zeta>0$. Roughly speaking, below we show that for sufficiently small $\zeta>0$---to be discussed later---``most'' matchings in the considered region have substantially more ``destabilizing'' blocking pairs than ``stabilizing'' ones. We then connect the dynamics to a random walk that is heavily biased in the ``destabilizing'' direction. This bias leads to exponentially long paths to stability.\bigskip

Consider an arbitrary matching $\lambda$ from the region of interest, i.e., $\mathcal{S}(\lambda) \geq (1-\zeta)n$. For the next few steps, suppose that $\lambda$ satisfies the following condition {\color{orange}($\star$)}: some agent $a \notin\{\bar{f}, \bar{w}\}$ is unmatched. This condition plays an important role; later, we show how to deal with matchings that do not satisfy the condition {\color{orange}($\star$)}.\bigskip

\noindent \textit{Types of Blocking Pairs.} For such a matching $\lambda$ satisfying {\color{orange}($\star$)}, we define two types of blocking pairs. A blocking pair is \textit{destabilizing\textsuperscript{$\star$}} if \begin{inparaenum}[(i)]
\item the corresponding new matching $\lambda'$, obtained from $\lambda$, also satisfies {\color{orange}($\star$)};
\item \textit{and, moreover,} the blocking pair itself is destabilizing, i.e., $\mathcal{S}(\lambda') = \mathcal{S}(\lambda)-1$.
\end{inparaenum}
In contrast, a blocking pair is \textit{stabilizing\textsuperscript{$\star$}} if \begin{inparaenum}[(i)]
\item the newly obtained matching $\lambda'$ does not satisfy {\color{orange}($\star$)};
\item \textit{or} the blocking pair itself is stabilizing, i.e., $\mathcal{S}(\lambda') = \mathcal{S}(\lambda)+1$.
\end{inparaenum} 

Note that for any other blocking pair that is neither destabilizing\textsuperscript{$\star$} nor stabilizing\textsuperscript{$\star$}, the newly obtained matching $\lambda'$ is ``similar'' to original $\lambda$. That is, $\mathcal{S}(\lambda') = \mathcal{S}(\lambda)$ and both new $\lambda'$ and original $\lambda$ satisfy {\color{orange}($\star$)}. Since we are interested in proving an exponential lower bound for the convergence time, we can essentially ``ignore'' such blocking pairs that only prolong the process by obtaining ``similar'' new $\lambda'$ from original $\lambda$.\bigskip

\noindent \textit{Number of (De)stabilizing\textsuperscript{$\star$} Blocking Pairs.}
We show that there are many more destabilizing\textsuperscript{$\star$} blocking pairs than stabilizing\textsuperscript{$\star$} ones, for sufficiently small $\zeta>0$. Indeed, by {\color{orange}($\star$)}, there is an unmatched agent $a \notin \{\bar{f}, \bar{w}\}$ that, by the premise of the proposition, is preferred by at least an $\eta$-proportion of agents from the other side to their stable partners. Since $\mathcal{S}(\lambda) \geq (1-\zeta)n$, i.e., at least a $(1-\zeta)$-proportion of agents are already matched with their stable partners,  there are at least $(\eta  -\zeta) n$ blocking pairs such that $\mathcal{S}(\lambda') = \mathcal{S}(\lambda)-1$. (\textit{Note:} The same holds true for best blocking pairs.) Among these blocking pairs, all pairs except possibly one pair necessarily unmatch some other $a' \notin \{\bar{f},\bar{w}\}$. Therefore, there are at least $(\eta-\zeta) n - 1$ destabilizing\textsuperscript{$\star$} blocking pairs.

As concern stabilizing\textsuperscript{$\star$} pairs, there are at most $\zeta n$ pairs such that $\mathcal{S}(\lambda') = \mathcal{S}(\lambda)+1$. Indeed, only firms that are not currently matched to their stable partners can be involved in such blocking pairs, and each such firm can take part in at most one such blocking pair. Furthermore, it is easy to check that there are at most two blocking pairs such that $\lambda'$ does not satisfy {\color{orange}($\star$)}; in that case,  either $\lambda'$ is perfect or only $\bar{f}$ and $\bar{w}$ are unmatched under $\lambda'$. To see this, consider the following three cases.
\begin{enumerate}
\item First, suppose that, under $\lambda$, agents $\bar{f}$ and $\bar{w}$ are either both matched or both unmatched. Then, $\lambda$ must have a pair of unmatched agents, $f \neq \bar{f}$ and $w \neq \bar{w}$. If $\lambda$ has several such pairs, $\lambda'$ must satisfy {\color{orange}($\star$)}. If $\lambda$ has only one such pair, $(f,w)$, then there is at most one blocking pair that makes $\lambda'$ violate {\color{orange}($\star$)}, i.e., the blocking pair $(f,w)$.

\item Second, suppose that $\bar{f}$ is unmatched and $\bar{w}$ is matched in $\lambda$. Therefore, $\lambda$ involves unmatched $w \neq \bar{w}$. If $\lambda$ has more unmatched pairs than just $(\bar{f},w)$, then $\lambda'$ must satisfy {\color{orange}($\star$)}. If $(\bar{f},w)$ is the only unmatched pair for $\lambda$, there are at most two blocking pairs that make $\lambda'$ violate {\color{orange}($\star$)}, the blocking pair $(\bar{f},w)$ that leads to perfect $\lambda'$ and the blocking pair $(\lambda(\bar{w}), w)$ that generates $(\bar{f}, \bar{w})$ as the only pair of unmatched agents.\item Third and finally, suppose that $\bar{f}$ is matched and $\bar{w}$ is unmatched in $\lambda$. As in the previous case, we can have at most two blocking pairs that make $\lambda'$ violate {\color{orange}($\star$)}.
\end{enumerate}
Thus, there are at most $\zeta n + 2$ stabilizing\textsuperscript{$\star$} blocking pairs.\bigskip

\noindent \textit{Likelihood of (De)stabilizing\textsuperscript{$\star$} Blocking Pairs.} The above analysis implies that for any unstable matching $\lambda$, satisfying {\color{orange}($\star$)} and belonging to the region $\mathcal{S}(\lambda) \geq (1-\zeta)n$, the ratio of the probability to satisfy a destabilizing\textsuperscript{$\star$} blocking pair to the probability to satisfy a stabilizing\textsuperscript{$\star$} blocking pair is bounded from below by 
{\setlength\abovedisplayskip{5pt}
\setlength\belowdisplayskip{5pt}
\begin{equation*}
\frac{(\eta-\zeta)n-1}{\kappa(\zeta n + 2)}  \geq \frac{\eta-\zeta}{2\kappa\zeta}\to \infty \text{ as } \zeta \to 0.   
\end{equation*}}

This, in turn, suggests that the probability to satisfy a destabilizing\textsuperscript{$\star$} blocking pair conditional on satisfying a destabilizing\textsuperscript{$\star$} or stabilizing\textsuperscript{$\star$} blocking pair---as already discussed, we essentially ``ignore'' all other blocking pairs---is bounded from below by
{\setlength\abovedisplayskip{5pt}
\setlength\belowdisplayskip{12pt}
\begin{equation*}
p_{d^\star} \equiv \frac{\eta-\zeta}{\eta+(2\kappa - 1)\zeta} \to 1 \text{ as } \zeta \to 0.
\end{equation*}}

In what follows, we discuss how to deal with matchings that do not satisfy the condition {\color{orange}($\star$)}. We then prove the desired result by establishing a close connection between our decentralized process and a biased random walk.\bigskip

\noindent \textit{Reinstatement of the Condition {\color{orange}($\star$)} for Matchings.} Consider a matching, say, $\lambda'$, that violates the condition {\color{orange}($\star$)}. Below, we demonstrate that after at most three additional steps---that is, after satisfying at most three blocking pairs in sequence---the desired condition {\color{orange}($\star$)} is restored for the eventually obtained matching;  assuming there is no path of length at most three from $\lambda'$ to the stable matching $\mu$. 
\begin{enumerate}
    \item If only $\bar{f}$ and $\bar{w}$ are unmatched in the given matching $\lambda'$, then in the next step, either $\lambda''$ satisfies {\color{orange}($\star$)}  or $\lambda''$ is perfect with $(\bar{f},\bar{w})$ being matched. In the latter case, after one more step, $\lambda'''$ must necessarily satisfy the condition {\color{orange}($\star$)}. Indeed, there is no way to unmatch only $\bar{f}$ and $\bar{w}$ again. 
    \item Otherwise, if $\lambda'$ is perfect, then in the next step, either $\lambda''$ satisfies {\color{orange}($\star$)} or only $\bar{f}$ and $\bar{w}$ are unmatched. In the latter case, by using similar arguments as before, we must obtain a matching satisfying {\color{orange}($\star$)} in at most two additional steps.
\end{enumerate}

\noindent \textit{Link to a Random Walk Heavily Based in the ``Destabilizing'' Direction.} We combine our previous observations to make several relaxations, still allowing us to obtain an exponential lower bound for the convergence time. Let the dynamics start at an unstable matching $\lambda$ with $\mathcal{S}(\lambda) = \ceil{(1-\zeta)n}$, for sufficiently small $\zeta>0$. 

Suppose that the starting matching $\lambda$ satisfies the condition {\color{orange}($\star$)}; if not, this condition must be restored after at most three steps. Then, if the dynamics select to satisfy a destabilizing\textsuperscript{$\star$} blocking pair, the newly obtained matching $\lambda'$ moves one stable pair away from stability, $\mathcal{S}(\lambda') = \mathcal{S}(\lambda)-1$, \textit{and} satisfies {\color{orange}($\star$)}. In contrast, when a stabilizing\textsuperscript{$\star$} blocking pair is satisfied, new $\lambda'$ either moves one stable pair towards stability, $\mathcal{S}(\lambda') = \mathcal{S}(\lambda)+1$, \textit{or} violates the condition {\color{orange}($\star$)}. Hereafter, we instead pessimistically assume---based on our analysis above---that after satisfying a stabilizing\textsuperscript{$\star$} blocking pair, new $\lambda'$ moves four stable pairs towards stability, $\mathcal{S}(\lambda') = \mathcal{S}(\lambda)+4$, \textit{and satisfies} the condition {\color{orange}($\star$)}. This relaxation effectively enables to preserve the condition {\color{orange}($\star$)} throughout the decentralized process.

To obtain the desired lower bound on the number of steps required for our dynamics to converge to the stable matching, $\mathcal{S}(\mu)=n$, we consider the biased random walk $\{S_i\}_{i\geq 0}$

\begin{center}
 \begin{tikzpicture}[bullet/.style={circle,inner sep=0.2ex},x=1.6cm,auto,bend angle=15]
\node[bullet, fill = black, label=below:{$k$}] (-3) at (-3,0) {};
\node[bullet, fill = black, label=below:{$k+1$}] (-2) at (-2,0) {};
\node[bullet, fill = black, label=below:{$k+2$}] (-1) at (-1,0) {};
\node[bullet, fill = black, label=below:{$k+3$}] (0) at (0,0) {};
\node[bullet, fill = black, label=below:{$k+4$}] (1) at (1,0) {};
\node[bullet, fill = black, label=below:{$k+5$}] (2) at (2,0) {};
\node[bullet, fill = black, label=below:{$n-1$}] (3) at (3,0) {};
\node[bullet, fill = black, label=below:{$\color{white}1\color{black}n\color{white}1$}] (4) at (4,0) {};
\node[bullet, fill = black, label=below:{$n+1$}] (5) at (5,0) {};
\node[bullet, fill = white, label=below:{}] (6) at (6,0) {};
  %(0,0) node[bullet,fill=red] (0) {}
  %(1,0) node[bullet,draw=blue] (a) {};
 %\foreach \Y [count=\X starting from -2] in {-2a,-a,0,a,2a} 
  %{\draw (\X,0.05) -- (\X,-0.1) node[below]{$\Y$};}
 \draw[-{Stealth[bend]}, color = black] (-3) to[bend left] node{w.p.\,1} (1);
 \draw[-{Stealth[bend]}, color = lightblue] (-2) to[bend left] node{} (2);
  \draw[-{Stealth[bend]}, dashed, color = lightblue] (-1) to[bend left] node{} (3);
    \draw[-{Stealth[bend]}, dashed, color = lightblue] (0) to[bend left] node{} (4);
      \draw[-{Stealth[bend]}, dashed, color = lightblue] (1) to[bend left] node{} (5);
       \draw[-{Stealth[bend]}, dashed, color = lightblue] (2) to[bend left] node{} (6);
 \draw[-{Stealth}, color = pear] (-2) -- (-3);
 \draw[-{Stealth}, color = pear] (-1) -- (-2);
  \draw[-{Stealth}, color = pear] (0) -- (-1);
    \draw[-{Stealth}, color = pear] (1) -- (0);
        \draw[-{Stealth}, color = pear] (2) -- (1);
        \draw[-{Stealth}, dashed, color = pear] (3) -- (2);
        \draw[-{Stealth}, color = pear] (4) -- (3);
        \draw[-{Stealth}, color = pear] (5) -- (4);
        \draw[-{Stealth}, dashed, color = pear] (6) -- (5);
 %\draw[-{Stealth[bend]},thick] (0) to[bend left] node{$1/2$} (-a);
\end{tikzpicture}  
\end{center}
on the set $\{k,k+1,k+2,\ldots\}$, with $k \equiv \ceil{(1-\zeta)n}$, defined as follows. This random walk starts at $S_0=k$. Furthermore, whenever $S_i=k$ for some $i\geq 0$, including starting $i=0$, the random walk deterministically moves four units to the right, $S_{i+1}=S_i+4$. Intuitively, if $\mathcal{S}(\lambda)$ ever drops below $k$ for our dynamics, we wait for it to return to that value and restart the random walk. If the ``starting'' matchings violate the condition {\color{orange}($\star$)}, this condition must be restored in at most, say, four steps---in fact, three steps as shown above.

When $S_i>k$ but $S_i<n$---this corresponds to our region of interest in which the condition {\color{orange}($\star$)} is preserved---the random walk moves one unit to the left, $S_{i+1}=S_i-1$, with probability
{\setlength\abovedisplayskip{5pt}
\setlength\belowdisplayskip{5pt}
\begin{equation*}
    {\color{pear}p_{d^\star}} = \frac{\eta-\zeta}{\eta+(2\kappa - 1)\zeta} \to 1 \text{ as } \zeta \to 0,
\end{equation*}} 

\noindent and four units to the right, $S_{i+1}=S_i+4$, with the remaining probability $ {\color{lightblue}p_{s^\star}} \equiv 1 - {\color{pear}p_{d^\star}} \to 0$ as $\zeta \to 0$. Transitions from the remaining nodes $\{n,n+1,n+2,\ldots\}$ are not relevant for our analysis and thus can be defined analogously, for instance.

To conclude, the number of steps required for the random walk $\{S_i\}_{i\geq 0}$ to first reach $S_i \geq n$ is a lower bound for the number of steps it takes for our dynamics to regain stability. The corresponding hitting time for such heavily biased random walks is known to be $2^{\Omega(n)}$ with probability $1-2^{-\Omega(n)}$; this result follows from standard Chernoff bound arguments, see Lemma 5 in the Online Appendix.
\end{proof}

\begin{delayedproof}{theorem_exponential}
\noindent Consider any sequence of \textit{original} markets of size $n \in \mathbb{N}$ and any $\delta>0$. For the purpose of augmentation, also take some arbitrary sequence of \textit{constructed} markets for which Proposition~\ref{proposition_class} holds; for instance, the sequence from the proof of Lemma 6 in the Online Appendix. Below, we use these constructed markets to generate the desired \textit{$\delta$-augmented} markets.

Specifically, for every given original market, of size $n$, add new agents with preferences over each other that replicate the constructed market of size $\floor{\delta n}$. In addition, let any agent---original or new---prefer original agents to new agents. It is straightforward to verify that the resulting market, of size $n+\floor{\delta n}$, is $\delta$-augmented. 

By the construction, the resulting sequence of $\delta$-augmented markets satisfies conditions similar to those in the statement of Proposition~\ref{proposition_class}: there exists $\eta \in (0,1)$ such that for every market in the sequence, any agent except one ``exclusion'' firm and one ``exclusion'' worker is preferred by at least an $\eta$-proportion of agents from the other side to their stable partners. 

The desired result for this sequence of $\delta$-augmented markets follows from the same arguments as in Proposition~\ref{proposition_class}. Although the market sizes in this sequence are $n+\floor{\delta n}$, not $n$, these sizes still increase linearly with $n \in \mathbb{N}$, so the analogous proof and conclusion hold.
\end{delayedproof}

\addcontentsline{toc}{section}{References}
\bibliographystyle{ecta}
\bibliography{ms}

\end{document}